\documentclass[10pt,english]{IEEEtran}
\usepackage[T1]{fontenc}
\usepackage[latin9]{inputenc}
\usepackage{array}
\usepackage{verbatim}
\usepackage{float}
\usepackage{fancybox}
\usepackage{calc}
\usepackage{multirow}
\usepackage{amsmath}
\usepackage{amssymb}
\usepackage{graphicx}

\makeatletter

\providecommand{\tabularnewline}{\\}
\floatstyle{ruled}
\newfloat{algorithm}{tbp}{loa}
\providecommand{\algorithmname}{Algorithm}
\floatname{algorithm}{\protect\algorithmname}

  \newtheorem{definitn}{Definition}
  \newtheorem{remrk}{Remark}
  \newtheorem{lemma}{Lemma}
  \newtheorem{thm}{Theorem}
  \newtheorem{example}{Example}

\usepackage{cite}

\author{Xiongbin~Rao,~\IEEEmembership{Student~Member,~IEEE} and~Vincent~K.~N.~Lau,~\IEEEmembership{Fellow,~IEEE}%\\
\thanks{The authors are with the Department of Electronic and Computer Engineering (ECE), the Hong Kong University of Science and Technology (HKUST), Hong Kong (e-mail: \{xrao,eeknlau\}@ust.hk).}%
 }

\makeatother

\usepackage{babel}
\begin{document}

\title{Compressive Sensing with Prior Support Quality Information and Application
to Massive MIMO Channel Estimation with Temporal Correlation}
\maketitle
\begin{abstract}
In this paper, we consider the problem of compressive sensing (CS)
recovery with a prior support and the prior support quality information
available. Different from classical works which exploit prior support
blindly, we shall propose novel CS recovery algorithms to exploit
the prior support adaptively based on the quality information. We
analyze the distortion bound of the recovered signal from the proposed
algorithm and we show that a better quality prior support can lead
to better CS recovery performance. We also show that the proposed
algorithm would converge in $\mathcal{O}\left(\log\mbox{SNR}\right)$
steps. To tolerate possible model mismatch, we further propose some
robustness designs to combat incorrect prior support quality information.
Finally, we apply the proposed framework to sparse channel estimation
in massive MIMO systems with temporal correlation to further reduce
the required pilot training overhead.
\end{abstract}

\section{Introduction}

The problem of recovering a sparse signal from a number of compressed
measurements has been drawing a lot of attention in the research community
\cite{foucart2013mathematical}. Specifically, consider the following
compressive sensing (CS) model:
\begin{equation}
\mathbf{y}=\mathbf{\Phi}\mathbf{x}\label{eq:signal_model}
\end{equation}
where $\mathbf{x}\in\mathbb{C}^{N\times1}$ is the unknown sparse
signal ($||\mathbf{x}||_{0}\ll N$), $\Phi\in\mathbb{C}^{M\times N}$
is the measurement matrix with $M\ll N$ , and $\mathbf{y}\in\mathbb{C}^{M\times1}$
are the measurements, where the goal is to recover $\mathbf{x}$ based
on $\mathbf{y}$ and $\mathbf{\Phi}$. Since $M\ll N$, (\ref{eq:signal_model})
is in fact an under-determined system and hence there are infinite
solutions of $\mathbf{x}$ to satisfy (\ref{eq:signal_model}) in
general. However, utilizing the fact that $\mathbf{x}$ is sparse
($||\mathbf{x}||_{0}\ll N$), it is possible recover $\mathbf{x}$
exactly via the following formulation \cite{foucart2013mathematical}:
\begin{equation}
\min_{\hat{\mathbf{x}}}||\hat{\mathbf{x}}||_{0}\quad\mbox{s.t.}\,\mathbf{y}=\mathbf{\Phi}\hat{\mathbf{x}}.\label{eq:l0_norm_mini}
\end{equation}

Unfortunately, problem (\ref{eq:l0_norm_mini}) is combinatorial and
has prohibitive complexity \cite{candes2005decoding}. To have feasible
solutions, researchers have designed many methods to approximately
solve (\ref{eq:l0_norm_mini}). For instance, the convex approximation
approach via $l_{1}$-norm minimization (basis pursuit) is proposed
in \cite{candes2005decoding}. Greedy-based algorithms which focus
on iteratively identifying the signal support (i.e.,$\mathcal{T}=\{i:\mathbf{x}(i)\neq0\}$)
or approximating the signal coefficients are proposed in \cite{tropp2007signal,blumensath2009iterative,needell2009cosamp,dai2009subspace}
(e.g., the orthogonal matching pursuit (OMP) in \cite{tropp2007signal},
iterative hard thresholding (IHT) in \cite{blumensath2009iterative},
compressive sampling matching pursuit (CoSaMP) in \cite{needell2009cosamp},
and subspace pursuit (SP) in \cite{dai2009subspace}). By using the
tools of the restricted isometry property (RIP) \cite{candes2005decoding},
these CS recovery algorithms \cite{tropp2007signal,blumensath2009iterative,needell2009cosamp,dai2009subspace,candes2005decoding}
are shown to achieve efficient recovery with substantially fewer measurements
compared with the signal dimension (i.e., $M\ll N$). Besides, there
are also works that deploy the approximate message passing technique
to achieve efficient CS recovery \cite{ziniel2013dynamic,vila2013expectation,ziniel2013efficient}.
However, they \cite{tropp2007signal,blumensath2009iterative,needell2009cosamp,dai2009subspace,candes2005decoding,vila2013expectation,ziniel2013dynamic,ziniel2013efficient}
consider one-time static CS recovery and do not exploit the prior
information of the signal support. 

In practice, we usually encounter the problem of recovering a sequence
of sparse signals and the sparse patterns of the signals are usually
correlated across time. For instance, consecutive real time video
signals \cite{wakin2006compressive,friedlander2012recovering} usually
have strong dependencies. In spectrum sensing, the index set of the
occupied frequency band usually varies slowly \cite{ren2012survey}.
In sparse channel estimation, consecutive frames tend to share some
multi-paths due to the slowly varying propagation environment between
base stations and users \cite{huang2009limited,bajwa2010compressed}.
As such, there is huge potential to exploit previously estimated signal
support to enhance the CS recovery performance at the present time.
In the literature, some works \cite{vaswani2010modified,jacques2010short,herzet2012exact,friedlander2012recovering,amel2014adaptive}
have already considered CS problems with a prior signal support $\mathcal{T}_{0}$
available and modified CS algorithms \cite{vaswani2010modified,jacques2010short,herzet2012exact,friedlander2012recovering,amel2014adaptive}
are proposed to exploit the prior $\mathcal{T}_{0}$ to enhance the
performance. For instance, in \cite{vaswani2010modified,jacques2010short,herzet2012exact},
modified basis pursuit designs are proposed to utilize $\mathcal{T}_{0}$
by %
\footnote{For instance, a typical modified $l_{1}$-norm minimization \cite{vaswani2010modified,jacques2010short,herzet2012exact,friedlander2012recovering,amel2014adaptive}
to exploit the prior support $\mathcal{T}_{0}$ is given by: $\min_{\hat{\mathbf{x}}}\left\Vert \hat{\mathbf{x}}_{\mathcal{T}_{0}^{c}}\right\Vert _{1}\quad\mbox{s.t.}\,\mathbf{y}=\mathbf{\Phi}\hat{\mathbf{x}}$. %
} minimizing the $l_{1}$-norm of the \emph{subvector} $\hat{\mathbf{x}}_{\mathcal{T}_{0}^{c}}$
formed by excluding the elements of $\mathbf{x}$ in $\mathcal{T}_{0}$,
$\mathcal{T}_{0}^{c}=\{1,...,N\}\backslash\mathcal{T}_{0}$. Based
on this, \cite{friedlander2012recovering,amel2014adaptive} have further
considered a \emph{weighted} $l_{1}$-norm minimization approach to
exploit $\mathcal{T}_{0}$. However, these designs \cite{vaswani2010modified,jacques2010short,herzet2012exact,friedlander2012recovering,amel2014adaptive}
do not take the \emph{quality} of the prior support information $\mathcal{T}_{0}$
into consideration in the problem formulation and fail to exploit
$\mathcal{T}_{0}$ \emph{adaptively} based on the quality%
\footnote{Here, the quality of prior support $\mathcal{T}_{0}$ refers to how
many indices in $\mathcal{T}_{0}$ are correct for the present. Please
refer to Section II for the details.%
} of $\mathcal{T}_{0}$. In practice, the prior signal support $\mathcal{T}_{0}$
may contain only \emph{part} of correct indices for the present time
(e.g., practical signal support is temporarily correlated but is also
\emph{dynamic} across time). In cases when only a small part of the
indices in $\mathcal{T}_{0}$ is correct, using the modified basis
pursuit design in \cite{vaswani2010modified,jacques2010short,herzet2012exact,friedlander2012recovering}
(which fully exploits $\mathcal{T}_{0}$), would lead to a even worse
performance \cite{friedlander2012recovering}. As such, it is desirable
to exploit $\mathcal{T}_{0}$ \emph{adaptively }based on how good
\emph{$\mathcal{T}_{0}$ }is for the present time.

In this paper, we propose a more complete model regarding the prior
signal support information. Aside from the prior support $\mathcal{T}_{0}$,
we assume that there is a metric to further indicate the \emph{quality}
of $\mathcal{T}_{0}$. Based on this, we design novel algorithms to
exploit $\mathcal{T}_{0}$ \emph{adaptively} based on the quality
indicator to achieve better signal recovery performance. Different
from previous works \cite{vaswani2010modified,jacques2010short,herzet2012exact,friedlander2012recovering,amel2014adaptive}
with convex relaxation approaches, we shall propose a greedy pursuit
approach%
\footnote{The focus of this work is on greedy pursuit based designs and the
detailed explanations for the selection the considered algorithm is
given at the beginning of Section III. Note that there may be other
approaches to exploit the prior support information, such as designing
from the approximate message passing \cite{ziniel2013dynamic,vila2013expectation,ziniel2013efficient}
which innately operates on the prior information of the signal. A
detailed investigation of other approaches is an interesting research
direction for future works. %
} to achieve our target. To cover more application scenarios, we shall
consider a framework with a general signal model which incorporates
conventional block sparsity \cite{eldar2010block,eldar2009robust}
and multiple measurement vector (MMV) joint sparsity models \cite{duarte2011structured,tropp2006algorithms,eldar2010average}.
There are several technical challenges to tackle in this work:
\begin{itemize}
\item \textbf{Algorithm Design to }\textbf{\emph{Adaptively}}\textbf{ Exploit
the Prior Support}: Note that classical CS works \cite{vaswani2010modified,jacques2010short,herzet2012exact,friedlander2012recovering,amel2014adaptive}
exploit prior support information $\mathcal{T}_{0}$ blindly. To further
enhance the recovery performance, we shall design a novel CS algorithm
to exploit the prior support $\mathcal{T}_{0}$ adaptively based on
the metric information indicating how good $\mathcal{T}_{0}$ is.
On the other hand, the proposed CS algorithm should also take the
general signal sparsity model into consideration. 
\item \textbf{Performance Analysis of the Proposed Algorithm}: Besides the
algorithm design, it is also important to quantify the performance
of the proposed novel CS recovery algorithms. For instance, it is
desirable to analyze the distortion bound of the recovered signal
and it is desirable to characterize the associated convergence speed
of the proposed algorithm.
\item \textbf{Robust Designs to Combat Model mismatch}: In practice, there
might be occasions with mismatch in the prior support information
model (e.g., \emph{incorrect} information of the prior support quality).
 For robustness, it it is also desirable to have some alternative
robust designs to make sure that the proposed scheme works efficiently
even with model mismatch.
\end{itemize}

In this paper, we shall address the above challenges. In Section II,
we introduce the CS problem setup with a general signal sparsity model.
We then present a prior support information model and introduce the
metric to quantify the \emph{quality} of the the prior support $\mathcal{T}_{0}$.
In Section III, we present the proposed CS algorithm to adaptively
exploit the prior support based on the quality indicator. After that,
in Section IV, we analyze the recovery performance of the proposed
algorithm, and in Section V, we further propose some robust designs
to tolerate model mismatch with incorrect prior support quality information.
Based on these results, in Section VI, we apply the proposed scheme
to sparse channel estimation in massive MIMO systems with temporal
correlation, to demonstrate the usefulness of the proposed framework.
Numerical results in Section VII demonstrate the performance advantages
of the the proposed scheme over the existing state-of-the-art algorithms.

\textit{Notation}s: Uppercase and lowercase boldface letters denote
matrices and vectors, respectively. The operators $(\cdot)^{T}$,
$(\cdot)^{*}$, $(\cdot)^{H}$, $(\cdot)^{\dagger}$, $|\cdot|$,
and $O(\cdot)$ are the transpose, conjugate, conjugate transpose,
Moore-Penrose pseudoinverse, cardinality, and big-O notation operator,
respectively; $\textrm{supp}(\mathbf{h})$ is the index set of the
non-zero entries of vector $\mathbf{h}$; $||\mathbf{A}||_{F}$, $||\mathbf{A}||$
and $||\mathbf{a}||$ denote the Frobenius norm, spectrum norm of
$\mathbf{A}$ and Euclidean norm of vector $\mathbf{a}$, respectively.

\section{System Model}

\subsection{Compressive Sensing Model}

Suppose we have compressed measurements $\mathbf{Y}\in\mathbb{C}^{M\times L}$
of an unknown sparse signal matrix $\mathbf{X}\in\mathbb{C}^{N\times L}$
given by 
\begin{equation}
\mathbf{Y}=\Phi\mathbf{X}+\mathbf{N}\label{eq:CS_signal_model}
\end{equation}
where $\Phi\in\mathbb{C}^{M\times N}$ ($M\ll N$) is the measurement
matrix and $\mathbf{N}\in\mathbb{C}^{N\times L}$ is the measurement
noise. Our target is to recover $\mathbf{X}$ based on $\mathbf{Y}$
and $\Phi$. Before we elaborate the recovery algorithm, we first
elaborate the considered signal sparsity model and the prior support
information for $\mathbf{X}$ in the following sections.

\subsection{Signal Sparsity Model}

Many works have considered CS problems with joint sparsity structures
in the literature. For instance, block sparsity is considered in \cite{eldar2010block,eldar2009robust}
in which the target sparse vector\emph{ }(i.e.,\emph{ $L=1$} in (\ref{eq:signal_model}))
has simultaneously zero or non-zero blocks with block size $d$. On
the other hand, the MMV problem is discussed in \cite{duarte2011structured,tropp2006algorithms,eldar2010average}
where the target sparse \emph{matrix} ($L>1$) has simultaneously
zero or non-zero rows. By exploiting the joint sparsity structures,
better recovery performance can be achieved compared with conventional
CS algorithms \cite{eldar2010block,eldar2009robust,duarte2011structured,tropp2006algorithms,eldar2010average}.
Motivated by these works, we shall consider a general sparsity model
for $\mathbf{X}$ in (\ref{eq:CS_signal_model}) so that conventional
block sparsity or MMV sparsity structure can be incorporated. Suppose
the sparse matrix $\mathbf{X}\in\mathbb{C}^{N\times L}$ ($N=Kd$
) is a concatenation of $K$ \emph{chunks }where each chunk is of
size $d\times L$ and has \emph{simultaneously} zero or non-zero entries.
Denote $\mathbf{X}[k]\in\mathbb{C}^{d\times L}$ as its $k$-th chunk
of $\mathbf{X}$ as in Figure \ref{fig:Illustration-of-CS_model},
i.e., 
\begin{equation}
\mathbf{X}\triangleq\left[\begin{array}{cc}
\mathbf{X}[1] & \in\mathbb{C}^{d\times L}\\
\mathbf{X}[2] & \in\mathbb{C}^{d\times L}\\
\vdots & \vdots\\
\mathbf{X}[K] & \in\mathbb{C}^{d\times L}
\end{array}\right]\in\mathbb{C}^{N\times L},\label{eq:signal_model1}
\end{equation}
Define the \emph{chunk support} $\mathcal{T}$ (with chunk size $d\times L$---assumed
throughout the paper) of $\mathbf{X}$ as
\begin{equation}
\mathcal{T}\triangleq\left\{ n:\;\left\Vert \mathbf{X}[n]\right\Vert _{F}>0,1\leq n\leq K\right\} .\label{eq:sparse_signal2}
\end{equation}
We formally have the following definition of chunk-sparse matrices.

\begin{figure}
\begin{centering}
\includegraphics[scale=0.36]{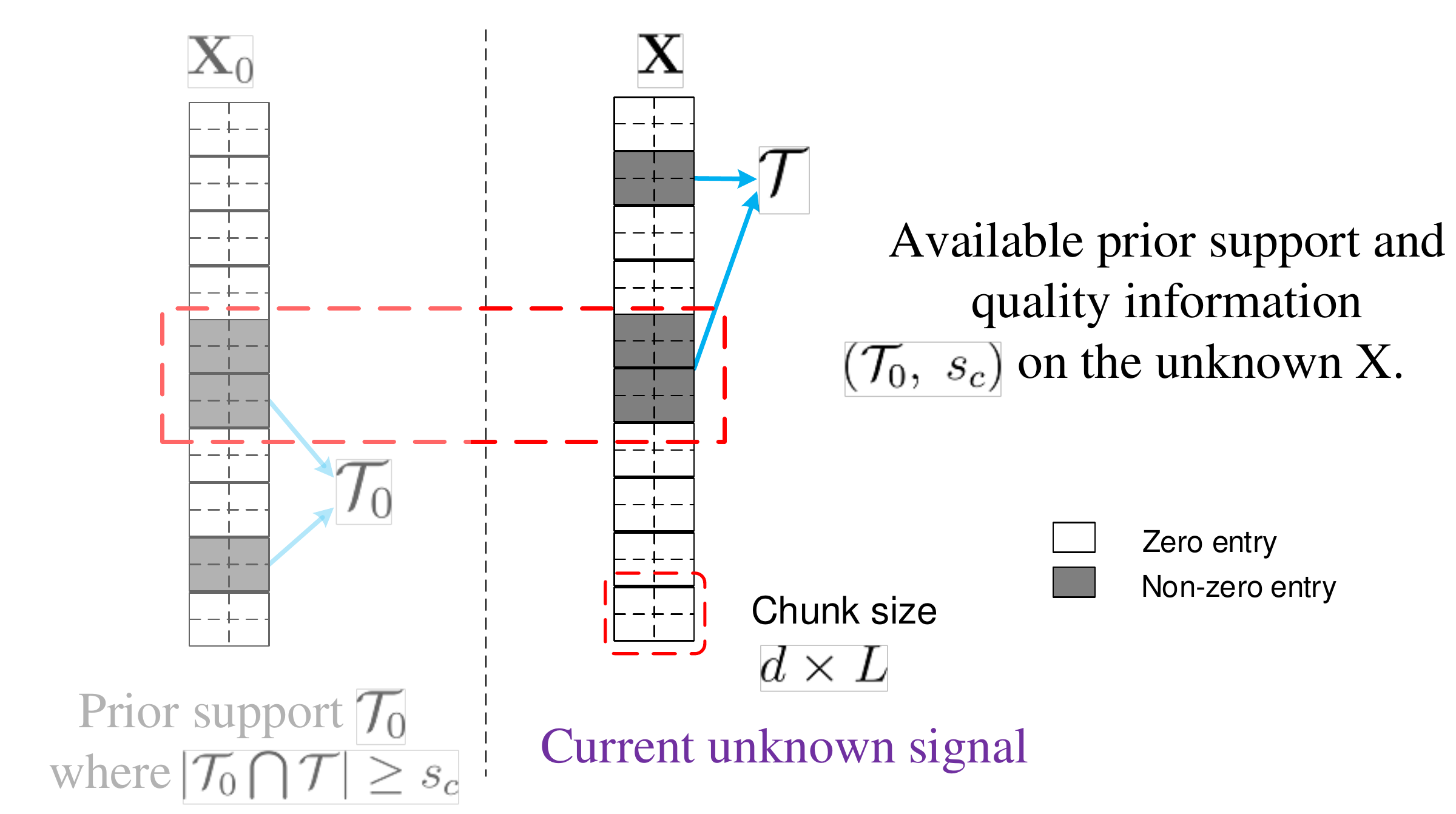}
\par\end{centering}

\protect\caption{\label{fig:Illustration-of-CS_model}Illustration of the available
prior support and quality information $(\mathcal{T}_{0},s_{c})$ for
$\mathbf{X}$. Note that $\mathcal{T}_{0}$ (available) and $\mathcal{T}$
(unknown) denote the prior support information and the current signal
support, respectively. Our target is to exploit the side information
of $(\mathcal{T}_{0},s_{c})$ to improve the CS recovery performance
of $\mathbf{X}$ from its compressed measurements $\mathbf{Y}$. }
\end{figure}

\begin{definitn}
[Chunk Sparsity Level]\label{Matrix-with-Chunk}Matrix $\mathbf{X}\in\mathbb{C}^{N\times L}$
is said to have $s$-th chunk sparsity level (CSL) if the chunk support
$\mathcal{T}$ of $\mathbf{X}$ as in (\ref{eq:sparse_signal2}) satisfies
$|\mathcal{T}|=s\ll K$. \hfill \QED
\end{definitn}

Note that when $d=1$ and $L>1$, the considered signal model is reduced
to the MMV joint sparsity models \cite{duarte2011structured,tropp2006algorithms,eldar2010average};
when $L=1$, the considered signal model is reduced to the block sparsity
scenarios \cite{eldar2010block,eldar2009robust}; and when both $d=1$
and $L=1$, the considered model degenerates to the classical signal
sparsity model without structures \cite{foucart2013mathematical}.
As such, the considered sparse signal model incorporates conventional
sparse signal models \cite{duarte2011structured,tropp2006algorithms,eldar2010average,eldar2010block,eldar2009robust}
and potentially covers more application scenarios. Note that practical
signals $\mathbf{X}$ may have some joint sparsity (e.g., due to physical
collocation \cite{barbotin2011estimation} or specific application
features such as Magnetic resonance imaging \cite{qiu2009real}) and
using a proper signal model enables us to exploit the joint sparsity
to enhance the signal recovery performance (as demonstrated in \cite{duarte2011structured,tropp2006algorithms,eldar2010average,eldar2010block,eldar2009robust}).
In this paper, we assume that the CSL of the target signal $\mathbf{X}\in\mathbb{C}^{N\times L}$
is upper bounded by $\bar{s}$, i.e., $|\mathcal{T}|\leq\bar{s}$
and $\bar{s}$ is available, as in classical CS works \cite{needell2009cosamp,dai2009subspace}.

\subsection{Prior Support Information}

We consider the following prior support information of $\mathbf{X}$
is available.
\begin{definitn}
[Prior Support Information]\label{Temporarily-Correlated-Sparse}The
\emph{prior support information} regarding the information $\mathbf{X}$
is characterized by a tuple $(\mathcal{T}_{0},s_{c})$, where $s_{c}\leq\left|\mathcal{T}_{0}\right|\leq\bar{s}$,
$\left|\mathcal{T}_{0}\bigcap\mathcal{T}\right|\geq s_{c}\geq0$.
\hfill \QED
\end{definitn}

\begin{remrk}
[Interpretation of Definition \ref{Temporarily-Correlated-Sparse}]Note
that $\mathcal{T}_{0}$ denotes the prior signal support and parameter
$s_{c}$ is a metric to indicate the \emph{quality} of the prior support
$\mathcal{T}_{0}$. Specifically, a larger $s_{c}$ means that a larger
number of indices in $\mathcal{T}_{0}$ is correct and hence means
a better quality of $\mathcal{T}_{0}$. Compared with conventional
works \cite{vaswani2010modified,jacques2010short,herzet2012exact,friedlander2012recovering,amel2014adaptive}
which exploit $\mathcal{T}_{0}$ blindly, we further consider some
uncertainty information about the prior support $\mathcal{T}_{0}$
(quantified by $s_{c}$) and such information allows us to exploit
$\mathcal{T}_{0}$ \emph{adaptively} based on $s_{c}$. Note that
$s_{c}$ refers to the number of correct indices but not the specific
indices in $\mathcal{T}_{0}\bigcap\mathcal{T}$. 
\end{remrk}

We then summarize the challenge we face in the following and we propose
a novel CS algorithm to handle the challenge in the next Section.\vspace{-0.4cm}

\begin{center}
\ovalbox{\begin{minipage}[t]{1\columnwidth}%
Challenge 1: Recover the chunk-sparse matrix $\mathbf{X}$ from $\mathbf{Y}$
in (\ref{eq:CS_signal_model}) exploiting the prior support information
$(\mathcal{T}_{0},\, s_{c})$.%
\end{minipage}}
\par\end{center}

\section{Algorithm Design to Exploit the Prior Support and Quality Information}

In this section, we shall propose a novel CS recovery algorithm to
solve Challenge 1 by \emph{extending} conventional greedy pursuit
algorithms with exploitation of $(\mathcal{T}_{0},\, s_{c})$ and
adaptation to the chunk sparsity structure of $\mathbf{X}$. Specifically,
we select to design from SP \cite{dai2009subspace} from the set of
greedy-based CS recovery algorithms, because SP \cite{dai2009subspace}
possesses many good properties such as uniform recovery guarantee
\cite{dai2009subspace}, relatively smaller required RIP constant
compared with other schemes of CoSaMP \cite{needell2009cosamp} or
IHT \cite{blumensath2009iterative} (based on the so far best known
RIP constants for these schemes \cite{giryes2012rip,chang2014improved}),
and closed-form characterizations on the number of iteration steps
\cite{dai2009subspace}. Hence, designing from SP \cite{dai2009subspace}
might enable us to obtain similar good properties. Moreover, the manipulations
of the support identification in SP \cite{dai2009subspace} provide
us an easy interface to incorporate the prior support information
$(\mathcal{T}_{0},\, s_{c})$. The detailed algorithm designs are
presented in the following.

\subsection{Algorithm Design}

\begin{table}
\begin{centering}
\begin{tabular}{|c|c|}
\hline 
\multirow{2}{*}{$\mathbf{x}^{\mathcal{T}}$ } & subvector formed by collecting the \tabularnewline
 & \emph{entries} of $\mathbf{x}$ indexed by $\mathcal{T}$.\tabularnewline
\hline 
\multirow{2}{*}{$\mathbf{X}^{[\mathcal{T}]}$} & submatrix formed by collecting the \tabularnewline
 & \emph{chunks }of\emph{ $\mathbf{X}$} indexed by $\mathcal{T}$.\tabularnewline
\hline 
\multirow{2}{*}{$\Phi_{\mathcal{T}}$} & submatrix formed by collecting the \tabularnewline
 & \emph{columns} of $\Phi$ indexed by $\mathcal{T}$.\tabularnewline
\hline 
\multirow{2}{*}{$\Phi_{[\mathcal{T}]}$} & submatrices formed by collecting \emph{columns}\tabularnewline
 &  of $\Phi$ indexed by $\{(k-1)d+1,..,kd:\forall k\in\mathcal{T}\}$. \tabularnewline
\hline 
\end{tabular}
\par\end{centering}

\protect\caption{\label{tab:Notations-rule}Notations.}
\end{table}

In \cite{dai2009subspace}, a subspace pursuit\emph{ }(SP)\emph{ }algorithm
is proposed to solve conventional CS problems. The basic idea of the
SP is to keep identifying the signal support based on the maximum
correlation criterion \cite{dai2009subspace} and by doing so, the
SP algorithm achieves efficient CS recovery with robustness to measurement
noises. In this section, we propose a \emph{modified subspace pursuit}
(M-SP) algorithm to solve Challenge 1 with exploitation of $(\mathcal{T}_{0},\, s_{c})$
and adaptation to the chunk sparsity model. To facilitate our presentations,
we first define a set of notation rules as in Table \ref{tab:Notations-rule}.
The details of the proposed M-SP algorithm are presented in Algorithm
\ref{alg:Modified-SP}.

\begin{algorithm}
\textbf{Input}: $\mathbf{Y}$, $\Phi$, $\bar{s}$, $(\mathcal{T}_{0},s_{c})$,
$\gamma$, $d$. 

\textbf{Output}: Estimated $\hat{\mathcal{T}}$ and $\hat{\mathbf{X}}$.

\textbf{Step 1} (\emph{Initialization}): Initialize the iteration
index $l=0$, chunk support$\hat{\mathcal{T}}_{l}=\emptyset$, and
the residue matrix $\mathbf{R}_{(l)}=\mathbf{Y}$. 

\textbf{Step 2} (\emph{Iteration}): Repeat the following steps until
stop.
\begin{itemize}
\item \textbf{A} (Support Merge): Set $\mathcal{T}_{a}=\hat{\mathcal{T}}_{l}\bigcup\left(\mathcal{T}_{b}\bigcup\mathcal{T}_{c}\right)$,
where{\small{}
\begin{equation}
\mathcal{T}_{b}=\arg\max_{|\mathcal{T}_{1}|=s_{c},\mathcal{T}_{1}\subseteq\mathcal{T}_{0}}\left\Vert \left(\Phi^{H}\mathbf{R}_{(l)}\right)^{[\mathcal{T}_{1}]}\right\Vert _{F}\label{eq:alg_equation1}
\end{equation}
\begin{equation}
\mathcal{T}_{c}=\arg\max_{|\mathcal{T}_{2}|=\bar{s}-s_{c},\mathcal{T}_{2}\subseteq\{1,..,K\}\backslash\mathcal{T}_{b}}\left\Vert \left(\Phi^{H}\mathbf{R}_{(l)}\right)^{[\mathcal{T}_{2}]}\right\Vert _{F}\label{eq:alg_equation2}
\end{equation}
}{\small \par}
\item \textbf{B} (\emph{LS Estimation})\emph{:} Set $\mathbf{Z}^{[\mathcal{T}_{a}]}=\Phi_{[\mathcal{T}_{a}]}^{\dagger}\mathbf{Y}$
and $\mathbf{Z}^{[\{1,...,K\}\backslash\mathcal{T}_{a}]}=\mathbf{0}$
.
\item \textbf{C} (\emph{Support Refinement}): Select $\hat{\mathcal{T}}_{l+1}$
as follows:{\small{}
\begin{align}
\hat{\mathcal{T}}_{l+1} & =\left\{ \arg\max_{|\mathcal{T}_{1}|=s_{c},\mathcal{T}_{1}\subseteq\mathcal{T}_{0}}\left\Vert \mathbf{Z}^{[\mathcal{T}_{1}]}\right\Vert _{F}\right\} \nonumber \\
 & \bigcup\left\{ \arg\max_{|\mathcal{T}_{2}|=\bar{s}-s_{c},\mathcal{T}_{2}\subseteq\{1,..,K\}\backslash\mathcal{T}_{1}}\left\Vert \mathbf{Z}^{[\mathcal{T}_{2}]}\right\Vert _{F}\right\} \label{eq:alg_equation3}
\end{align}
}{\small \par}
\item \textbf{D} (\emph{Signal Estimation}): Set $\hat{\mathbf{X}}_{(l+1)}^{\left[\hat{\mathcal{T}}_{l+1}\right]}=\Phi_{\left[\hat{\mathcal{T}}_{l+1}\right]}^{\dagger}\mathbf{Y}$
and $\hat{\mathbf{X}}_{(l+1)}^{[\{1,...,K\}\backslash\hat{\mathcal{T}}]}=\mathbf{0}$. 
\item \textbf{E} (\emph{Residue}): Compute $\mathbf{R}_{(l+1)}=\mathbf{Y}-\Phi_{\left[\hat{\mathcal{T}}_{l+1}\right]}\hat{\mathbf{X}}_{(l+1)}^{\left[\hat{\mathcal{T}}_{l+1}\right]}$.
\item \textbf{F} \emph{(Stopping Condition and Output)}: If $\left\Vert \mathbf{R}_{(l+1)}\right\Vert _{F}\leq\gamma$,
stop and output $\hat{\mathcal{T}}=\hat{\mathcal{T}}_{l+1}$ and $\hat{\mathbf{X}}=\hat{\mathbf{X}}_{(l+1)}$;
Else if $\left\Vert \mathbf{R}_{(l+1)}\right\Vert _{F}\geq\left\Vert \mathbf{R}_{(l)}\right\Vert _{F}$,
stop and output $\hat{\mathcal{T}}=\hat{\mathcal{T}}_{l}$ and $\hat{\mathbf{X}}=\hat{\mathbf{X}}_{(l)}$;
Else, set $l=l+1$ and go to \textbf{Step 2 A}. 
\end{itemize}

\protect\caption{\label{alg:Modified-SP}Modified-SP to Solve Challenge 1.}
\end{algorithm}

\begin{remrk}
[Interpretation of Algorithm \ref{alg:Modified-SP}]In the proposed
M-SP algorithm (Algorithm \ref{alg:Modified-SP}), $\gamma$ is a
threshold parameter, $\hat{\mathcal{T}}_{l}$ and $\hat{\mathbf{X}}_{(l)}$
denote the estimated chunk support and the estimated signal for $\mathbf{X}$
in the $l$-th iteration, respectively. Note that when $d=1$, $L=1$
and $s_{c}=0$, Algorithm \ref{alg:Modified-SP} will degenerate to
conventional SP \cite{dai2009subspace} (except that the M-SP has
a different stopping criterion%
\footnote{\label{fn:stopping}Note that the more sophisticated stopping conditions
in the M-SP algorithm (compared with that in conventional SP \cite{dai2009subspace})
enable us to obtain more complete convergence results. For instance,
as illustrated in Table \ref{tab:Comparison-between-SP-MSP}, conventional
SP \cite{dai2009subspace} only characterizes the number of convergence
steps in noise free cases while our results cover both noise-free
and noisy scenarios. %
}. Table \ref{tab:Comparison-between-SP-MSP} illustrates the comparison
between conventional SP and the proposed M-SP. The following explains
how the proposed M-SP exploits the prior support information $(\mathcal{T}_{0},s_{c})$
and adapts to the chunk sparsity model:\end{remrk}
\begin{itemize}
\item \textbf{Exploitation of Prior Support Information} $(\mathcal{T}_{0},s_{c})$
: Note that the information $(\mathcal{T}_{0},s_{c})$ is utilized
in \textbf{Step 2 A} and \textbf{C} of Algorithm \ref{alg:Modified-SP}.
As can be seen in \textbf{Step 2 A}, the newly added support (i.e.,
$\mathcal{T}_{b}\bigcup\mathcal{T}_{c}$) contains two parts, namely
$\mathcal{T}_{b}$ and $\mathcal{T}_{c}$, where $\mathcal{T}_{b}$
with size $s_{c}$ is selected from prior support $\mathcal{T}_{0}$,
$\mathcal{T}_{c}$ with size $\bar{s}-s_{c}$ is selected from $\{1,..,K\}\backslash\mathcal{T}_{b}$.
This design utilizes the fact that prior support $\mathcal{T}_{0}$
contains at least $s_{c}$ correct indices. Similarly, in \textbf{Step
2 C}, the refined signal support $\hat{\mathcal{T}}_{l+1}$ contains
two parts, i.e., $s_{c}$ indices from $\mathcal{T}_{0}$ and another
$\bar{s}-s_{c}$ from the others. This is in contrast to conventional
SP \cite{dai2009subspace} in which the newly added/updated signal
support are \emph{blindly} selected\emph{ over the entire} signal
index set $\{1,..,K\}$. Using the proposed support identification
criterion, the prior support information $\mathcal{T}_{0}$ is utilized
adaptively based on the quality information $s_{c}$, and hence better
recovery performance may be achieved. 
\item \textbf{Adaptation to the General Sparsity Model}: Note that we have
considered a general sparsity model in which the signal matrix $\mathbf{X}$
has simultaneous zero or non-zero entries within each chunk (with
size $d\times L$). Therefore, instead of identifying each single
element in $\mathbf{X}$ separately (as in the conventional SP \cite{dai2009subspace}),
we identify each non-zero chunk as a atomic unit based on the aggregate
correlation effects between the measurement matrix $\Phi$ and the
residue matrix $\mathbf{R}_{(l)}$. For instance, in (\ref{eq:alg_equation1})-(\ref{eq:alg_equation2}),
we identify a new chunk based on the Frobenius norm of $\left(\Phi^{H}\mathbf{R}_{(l)}\right)^{[\{k\}]}$
which corresponds to an aggregate correlation effect due to the $k$-th
chunk. This design adapts to the joint sparsity structure in $\mathbf{X}$
and may achieve better recovery performance \cite{duarte2011structured,tropp2006algorithms,eldar2010average,eldar2010block,eldar2009robust}.
\end{itemize}

After giving the details of the proposed M-SP algorithm above, it
is also important for us to characterize the associated recovery performance.
Specifically, we are interested in characterizing the distortion bound
of the estimated signal as well as the convergence speed of the proposed
M-SP algorithm. We shall discuss these issues in the next Section.

\begin{table}
\begin{centering}
\includegraphics[scale=0.55]{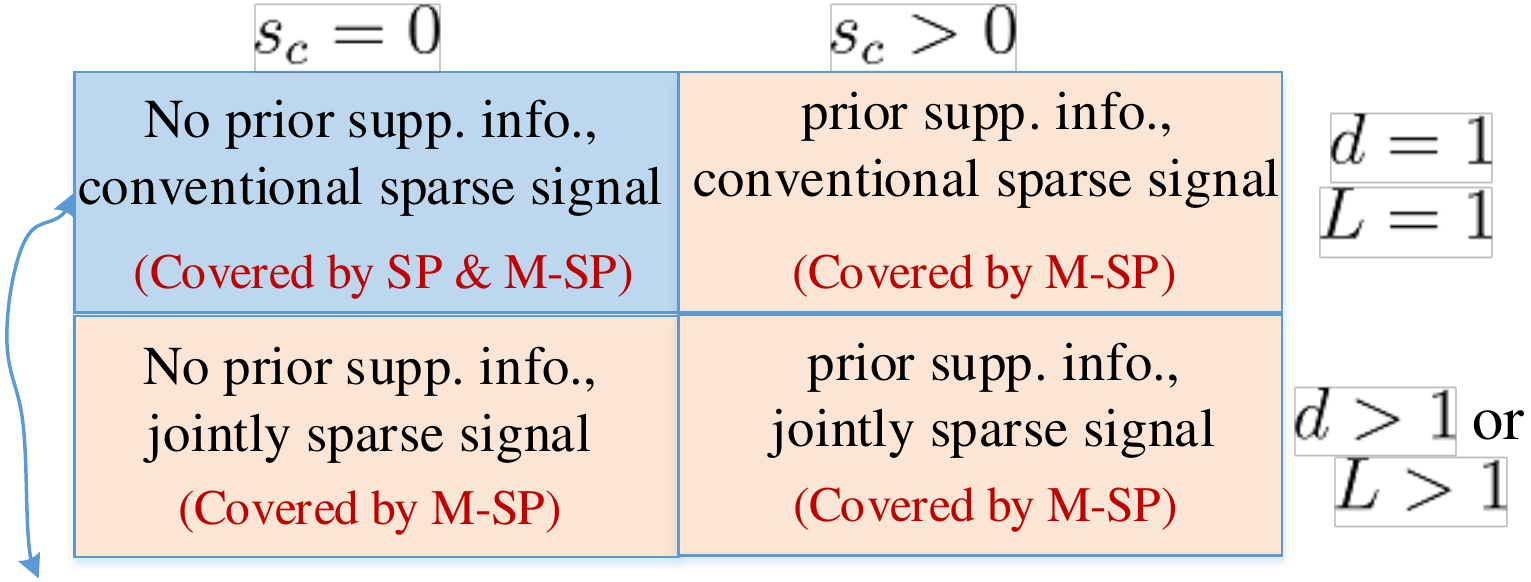}
\par\end{centering}

\begin{centering}
{\small{}}%
\begin{tabular}{|c|c|c|c|}
\hline 
\multicolumn{2}{|c|}{{\small{}Comparison ($s_{c}=0$, $d=1$, $L=1$)}} & {\small{}SP \cite{dai2009subspace,chang2014improved}} & {\small{}M-SP}\tabularnewline
\hline 
\hline 
 & {\small{}Performance} & \multirow{4}{*}{{\small{}covered}} & \multirow{8}{*}{{\small{}covered}}\tabularnewline
{\small{}no noise} & {\small{}i.e., $\mathbf{x}=\hat{\mathbf{x}}$} &  & \tabularnewline
\cline{2-2} 
{\small{}$\mathbf{y}=\Phi\mathbf{x}$} & {\small{}Convergence} &  & \tabularnewline
 & {\small{}\# iterations $n_{co}$} &  & \tabularnewline
\cline{1-3} 
 & {\small{}Performance, i.e.,} & \multirow{2}{*}{{\small{}covered}} & \tabularnewline
{\small{}noisy } & {\footnotesize{}$||\mathbf{x}-\hat{\mathbf{x}}||\leq\mathcal{O}\left(||\mathbf{n}||\right)$} &  & \tabularnewline
\cline{2-3} 
{\small{}$\mathbf{y}=\Phi\mathbf{x}+\mathbf{n}$} & {\small{}Convergence} & \multirow{2}{*}{{\small{}not covered}} & \tabularnewline
 & {\small{}\# iterations $n_{co}$} &  & \tabularnewline
\hline 
\end{tabular}
\par\end{centering}{\small \par}

\protect\caption{\label{tab:Comparison-between-SP-MSP}Comparison of the proposed M-SP
and SP {\small{}\cite{dai2009subspace,chang2014improved,giryes2012rip}}.}
\end{table}

\section{Performance Analysis of the Proposed M-SP}

In this Section, we analyze the performance of the proposed M-SP algorithm
by deploying the tools of restricted isometry property (RIP) \cite{candes2005decoding,eldar2009robust}.
Specifically, we are interested in both the estimation distortion
(e.g., $\left\Vert \mathbf{X}-\hat{\mathbf{X}}\right\Vert _{F}$)
and the convergence speed of Algorithm \ref{alg:Modified-SP}. Based
on the results, we further derive some simple insights regarding how
the prior support quality $s_{c}$ affects the recovery performance. 

\begin{center}
\vspace{-0.4cm}
\ovalbox{\begin{minipage}[t]{1\columnwidth}%
Challenge 2: Analyze the distortion of the estimated signal $\hat{\mathbf{X}}$
and the associated convergence speed for Algorithm \ref{alg:Modified-SP}.%
\end{minipage}}
\par\end{center}

\subsection{Preliminaries}

In the literature, the RIP \cite{candes2005decoding} is commonly
adopted to facilitate the performance of CS recovery algorithms. However,
the conventional RIP \cite{candes2005decoding} only serves to handle
general sparse signal vectors without sparsity structures. To deal
with the CS problems with block sparsity structures, the authors in
\cite{eldar2009robust} further propose the notion of \emph{block-RIP}
by extending the conventional RIP \cite{candes2005decoding}. This
block RIP \cite{eldar2009robust} can also be deployed to facilitate
the performance analysis in our scenario. We first review the notion
of the block-RIP \cite{eldar2009robust} as follows:
\begin{definitn}
[Block Restricted Isometry Property \cite{eldar2009robust}]\label{Block-Restricted-Isometry}Matrix
$\Phi\in\mathbb{C}^{M\times N}$ satisfies the $k$-th order block-RIP
with block size $d$ ($d\mid N$, $K\triangleq\frac{N}{d}$) and block-RIP
constant $\delta_{k|d}$, if $0\leq\delta_{k|d}<1$ and
\begin{align*}
\delta_{k|d}\triangleq & \min\left\{ \delta:\;(1-\delta)\left\Vert \mathbf{x}\right\Vert _{2}^{2}\leq\left\Vert \Phi\mathbf{x}\right\Vert _{2}^{2}\right.\\
 & \left.\leq(1+\delta)\left\Vert \mathbf{x}\right\Vert _{2}^{2},\;|\textrm{supp}_{d}(\mathbf{x})|\leq k\right\} 
\end{align*}
where $\textrm{supp}_{d}(\mathbf{x})=\left\{ n:||\mathbf{x}[n]||>0,1\leq n\leq K\right\} $
with $\mathbf{x}[n]$ denoting the $n$-th block of $\mathbf{x}$
(with block size $d\times1$) \cite{eldar2009robust}.\hfill \QED
\end{definitn}

Note that when $d=1$, the block-RIP will be reduced to the conventional
RIP \cite{candes2005decoding}. In the following analysis, we assume
that the measurement matrix $\Phi$ has block-RIP properties with
$\delta_{k|d}$ denoting the $k$-th order block-RIP constant of $\Phi$.
We first introduce the following inequalities over the block-RIP by
extending conventional results \cite{needell2009cosamp,dai2009subspace,chang2014improved}.
\begin{lemma}
[Inequalities over the block-RIP]\label{Inequalities-Over-Block-RIP}The
following inequalities are satisfied:

1) If $k_{1}\leq k_{2}$, then $\delta_{k_{1}|d}\leq\delta_{k_{2}|d}$. 

2) For support $\mathcal{T}$ with $\left|\mathcal{T}\right|\leq k$,
we have 
\begin{eqnarray}
1-\delta_{k|d} & \leq & \sigma_{\min}\left(\Phi_{[\mathcal{T}]}^{H}\Phi_{[\mathcal{T}]}\right)\label{eq:inequality1}\\
 &  & \leq\sigma_{\max}\left(\Phi_{[\mathcal{T}]}^{H}\Phi_{[\mathcal{T}]}\right)\leq1+\delta_{k|d},\nonumber 
\end{eqnarray}
\begin{equation}
\sigma_{\max}\left(\Phi_{[\mathcal{T}]}^{\dagger}\right)\leq\frac{1}{\sqrt{1-\delta_{k|d}}}.\label{eq:inequality4}
\end{equation}

3) For two disjoint supports $\mathcal{T}_{1}$, $\mathcal{T}_{2}$,
where $\left|\mathcal{T}_{1}\right|\leq k_{1}$, $\left|\mathcal{T}_{2}\right|\leq k_{2}$,
$\mathcal{T}_{1}\bigcap\mathcal{T}_{2}=\emptyset$, we have

\begin{equation}
\sigma_{\max}\left(\Phi_{[\mathcal{T}_{1}]}^{H}\Phi_{[\mathcal{T}_{2}]}\right)\leq\delta_{k_{1}+k_{2}|d}.\label{eq:inequality5}
\end{equation}

4) Suppose the chunk support of $\mathbf{X}$ is $\mathcal{T}_{1}$.
Suppose$\mathcal{T}_{1}$, $\mathcal{T}_{2}$ are two disjoint supports
where $\left|\mathcal{T}_{1}\right|\leq k_{1}$, $\left|\mathcal{T}_{2}\right|\leq k_{2}$,
$\mathcal{T}_{1}\bigcap\mathcal{T}_{2}=\emptyset$. Denote the projection
matrix $\mathbf{P}_{(\mathcal{T}_{2})}$ as $\mathbf{P}_{(\mathcal{T}_{2})}\triangleq\Phi_{[\mathcal{T}_{2}]}\left(\Phi_{[\mathcal{T}_{2}]}^{H}\Phi_{[\mathcal{T}_{2}]}\right)^{-1}\Phi_{[\mathcal{T}_{2}]}^{H}$.
Then, 
\[
\left\Vert \mathbf{P}_{(\mathcal{T}_{2})}\Phi\mathbf{X}\right\Vert _{F}\leq\delta_{k_{1}+k_{2}|d}\sqrt{1+\delta_{k_{1}+k_{2}|d}}\left\Vert \mathbf{X}\right\Vert _{F}.
\]
\end{lemma}
\begin{proof}
See Appendix \ref{sub:Proof-of-Lemma-property}.
\end{proof}

\subsection{Performance Analysis of the Proposed M-SP}

Using the properties in Lemma \ref{Inequalities-Over-Block-RIP},
we obtain the following property regarding the residue matrix $\mathbf{R}_{(l+1)}$
and estimated signal $\hat{\mathbf{X}}_{(l+1)}$ in the $l$-th iteration
of Algorithm \ref{alg:Modified-SP}. 
\begin{lemma}
[Iteration Property in Algorithm \ref{alg:Modified-SP}]\label{Iteration-Property-II}In
the $l$-th iteration $(l\geq1)$ in Step 2 of Algorithm \ref{alg:Modified-SP},
the following inequalities are satisfied:{\small{}
\begin{equation}
\left\Vert \mathbf{R}_{(l+1)}\right\Vert _{F}\leq C_{1}\left\Vert \mathbf{R}_{(l)}\right\Vert _{F}+C_{2}\eta\label{eq:Iteration_Lemma}
\end{equation}
\begin{equation}
\left\Vert \mathbf{X}-\hat{\mathbf{X}}_{(l+1)}\right\Vert _{F}\leq(C_{1})^{l+1}\sqrt{\frac{1+\delta_{\bar{s}|d}}{1-\delta_{s_{1}|d}}}\left\Vert \mathbf{X}\right\Vert _{F}+C_{3}(l)\eta\label{eq:first_distortion}
\end{equation}
}where $\eta=\left\Vert \mathbf{N}\right\Vert _{F}$ is the noise
magnitude, $C_{1}$, $C_{2}$ and $C_{3}(l)$ are expressed in Table
\ref{tab:The-detailed-expressions-constants}.\end{lemma}
\begin{proof}
See Appendix \ref{sub:Proof-of-key-lemma}.
\end{proof}

\begin{table}
\begin{centering}
{\small{}}%
\begin{tabular}{|c|c|}
\hline 
{\small{}$C_{1}$} & {\small{}$\frac{2\delta_{s_{2}|d}\sqrt{1+\delta_{s_{2}|d}}\sqrt{1-\delta_{s_{2}|d}+4\delta_{s_{2}|d}^{2}+4\delta_{s_{2}|d}^{3}}}{\left(1-\delta_{s_{2}|d}\right)^{2}}$}\tabularnewline
\hline 
\multirow{2}{*}{{\small{}$C_{2}$}} & {\small{}$2\sqrt{\frac{1+\delta_{\bar{s}|d}}{1-\delta_{s_{1}|d}}}+\sqrt{1+\delta_{\bar{s}|d}}\sqrt{1+\frac{4\delta_{s_{2}|d}^{2}(1+\delta_{s_{2}|d})}{1-\delta_{s_{1}|d}}}$}\tabularnewline
 & {\small{}$\times\left(\frac{2\delta_{s_{2}|d}}{\left(1-\delta_{\bar{s}|d}\right)\sqrt{1-\delta_{s_{1}|d}}}+\frac{2\sqrt{1+\delta_{\bar{s}|d}}}{1-\delta_{\bar{s}|d}}\right)+1$}\tabularnewline
\hline 
$C_{3}(l)$ & $\frac{(C_{1})^{l}\left(1-\frac{C_{2}}{1-C_{1}}\right)+\frac{C_{2}}{1-C_{1}}+1}{\sqrt{1-\delta_{s_{1}|d}}}$\tabularnewline
\hline 
{\small{}$C_{4}$} & {\small{}$\frac{(1-C_{1}+C_{2})}{(1-C_{1})\sqrt{1-\delta_{s_{1}|d}}}$}\tabularnewline
\hline 
{\small{}where} & {\small{}$s_{1}\triangleq2\bar{s}+\min\left(0,\,\left|\mathcal{T}_{0}\right|-2s_{c}\right)$}\tabularnewline
 & {\small{}$s_{2}\triangleq3\bar{s}+\min\left(0,\,\left|\mathcal{T}_{0}\right|-3s_{c}\right)$}\tabularnewline
\hline 
\end{tabular}
\par\end{centering}{\small \par}

\protect\caption{\label{tab:The-detailed-expressions-constants}The detailed expressions
for constants.}

\end{table}

\begin{comment}
\begin{align*}
C_{1} & =\frac{4\delta_{s_{2}|d}\sqrt{1+\delta_{\bar{s}|d}}}{\left(1-\delta_{\bar{s}|d}\right)\sqrt{1-\delta_{s_{1}|d}}}\left(1+\frac{\delta_{s_{2}|d}}{1-\delta_{s_{1}|d}}\right)
\end{align*}
\begin{align*}
C_{2} & =1+2\sqrt{\frac{1+\delta_{\bar{s}|d}}{1-\delta_{s_{1}|d}}}+4\sqrt{1+\delta_{\bar{s}|d}}\left(1+\frac{\delta_{s_{2}|d}}{1-\delta_{s_{1}|d}}\right)\times\\
 & \left(\frac{2\delta_{s_{2}|d}}{\left(1-\delta_{\bar{s}|d}\right)\sqrt{1-\delta_{s_{1}|d}}}+\frac{2\sqrt{1+\delta_{\bar{s}|d}}}{1-\delta_{\bar{s}|d}}\right)
\end{align*}
\[
C_{2}=1+2\sqrt{\frac{1+\delta_{\bar{s}|d}}{1-\delta_{s_{1}|d}}}+4\sqrt{1+\delta_{\bar{s}|d}}\left(1+\frac{\delta_{s_{2}|d}}{1-\delta_{s_{1}|d}}\right)\left(\frac{2\delta_{s_{2}|d}}{\left(1-\delta_{\bar{s}|d}\right)\sqrt{1-\delta_{s_{1}|d}}}+\frac{2\sqrt{1+\delta_{\bar{s}|d}}}{1-\delta_{\bar{s}|d}}\right)
\]
\end{comment}

Note that equations (\ref{eq:Iteration_Lemma})-(\ref{eq:first_distortion})
in Lemma \ref{Iteration-Property-II} are very important to derive
the distortion bound and the convergence speed in Theorem \ref{Distortion-BoundDenote-MSP}
and \ref{Convergence-SpeedSuppose-II} respectively. For instance,
if $C_{1}<1$, then the distortion (i.e., {\small{}$\left\Vert \mathbf{X}-\hat{\mathbf{X}}_{(l)}\right\Vert _{F}$}
in (\ref{eq:first_distortion})) turns to decrease exponentially with
ratio $C_{1}$ in the iterations of the proposed M-SP. Based on Lemma
\ref{Iteration-Property-II}, we obtain the following recovery distortion
bound for the proposed M-SP algorithm.

\begin{thm}
[Recovery Distortion Bound]\label{Distortion-BoundDenote-MSP}Suppose
the $s_{2}$-th block RIP constant $\delta_{s_{2}|d}$ satisfies $\delta_{s_{2}|d}<0.246$.
The following properties are true regarding Algorithm \ref{alg:Modified-SP}:

(i) The final obtained solution $\mathbf{\hat{X}}$ satisfies

\begin{equation}
\left\Vert \mathbf{X}-\hat{\mathbf{X}}\right\Vert _{F}\leq\max\left(C_{4}\eta,\quad\frac{\gamma+\eta}{\sqrt{1-\delta_{s_{1}|d}}}\right).\label{eq:Distortion_Bound}
\end{equation}

(ii) If the signal $\mathbf{X}$ satisfies{\small{} $\min_{k\in\mathcal{T}}\left\Vert \mathbf{X}[k]\right\Vert _{F}>\max\left(C_{4}\eta,\quad\frac{\gamma+\eta}{\sqrt{1-\delta_{s_{1}|d}}}\right)$},
then final obtained solution $\mathbf{\hat{X}}$ further satisfies
\begin{equation}
\left\Vert \mathbf{X}-\hat{\mathbf{X}}\right\Vert _{F}\leq\frac{1}{\sqrt{1-\delta_{\bar{s}|d}}}\eta\label{eq:refined_bound}
\end{equation}
where $\gamma$ is the threshold parameter in Algorithm \ref{alg:Modified-SP},
$s_{1}$, $s_{2}$ and $C_{4}$ are given in Table \ref{sub:Proof-of-key-lemma}. \end{thm}
\begin{proof}
See Appendix \ref{sub:Proof-of-Theorem-Distortion_Bound}.
\end{proof}

\begin{remrk}
[Interpretation of Theorem \ref{Distortion-BoundDenote-MSP}]Note
that $\delta_{s_{2}|d}<0.246$ is to ensure $C_{1}<1$ in (\ref{eq:Iteration_Lemma})-(\ref{eq:first_distortion}).
When there is no noise in the system, i.e., $\eta\triangleq\left\Vert \mathbf{N}\right\Vert _{F}=0$
and $\gamma$ is set to be $\gamma=0$ , then perfect signal recovery,
i.e., $\mathbf{X}=\hat{\mathbf{X}}$, will be achieved from (\ref{eq:Distortion_Bound}).
Based on Theorem \ref{Distortion-BoundDenote-MSP}, we have the following
discussion regarding the proposed M-SP algorithm:\end{remrk}
\begin{itemize}
\item \textbf{Backward Compatibility with Conventional SP} \cite{dai2009subspace,chang2014improved}:
Note that when $d=1$, $L=1$ and $s_{c}=0$, the proposed M-SP will
be reduced to the conventional SP \cite{dai2009subspace} (except
that we have more sophisticated stopping conditions as explained in
footnote \ref{fn:stopping}). In such a scenario, the requirement
on the RIP constant in Theorem \ref{Distortion-BoundDenote-MSP} becomes
$\delta_{3\bar{s}}<0.246$, which is slightly better (a slightly weaker
requirement) than the so-far best known bound ($\delta_{3\bar{s}}<0.2412$)
derived for SP in Thm. 3.8 of \cite{chang2014improved} . This is
because we have combined the techniques in these pioneering works
\cite{needell2009cosamp,chang2014improved,dai2009subspace}, to derive
Lemma \ref{Iteration-Property-II} and Theorem \ref{Distortion-BoundDenote-MSP}
(please refer to Appendix \ref{sub:Proof-of-key-lemma} for the details).
A detailed comparison between the proposed M-SP and conventional SP
\cite{dai2009subspace} is given in Table \ref{tab:Comparison-between-SP-MSP}. 
\item \textbf{How Prior Support Quality $s_{c}$ Affects Performance}: From
Theorem \ref{Distortion-BoundDenote-MSP}, a larger $s_{c}$ (a higher
quality of the prior support $\mathcal{T}_{0}$) would achieve a better
CS recovery performance. For instance, suppose $\Phi$ is a i.i.d.
sub-Gaussian random%
\footnote{Note that the randomized approach is a commonly adopted method to
generate the CS measurement matrix for a good RIP property \cite{eldar2009robust}. %
} matrix \cite{eldar2009robust}, from \cite{eldar2009robust}, the
number of measurements $M$ to achieve the $s_{2}$-th order block-RIP
with $\delta_{s_{2}}=\delta$, is given by $M=\mathcal{O}(s_{2}d\ln\delta^{-1}+\delta^{-1}s_{2}\log K)$.
On the other hand, $s_{2}\triangleq3\bar{s}+\min\left(0,\,\left|\mathcal{T}_{0}\right|-3s_{c}\right)$
is monotonically decreasing as $s_{c}$ increases when $\frac{\left|\mathcal{T}_{0}\right|}{3}\leq s_{c}\leq\left|\mathcal{T}_{0}\right|$.
Therefore, a larger $s_{c}$ would lead to a weaker requirement on
the number of measurements $M$ to achieve the desired performance
in Theorem \ref{Distortion-BoundDenote-MSP}. 
\end{itemize}

We further have the following result regarding the convergence speed
of the proposed M-SP algorithm (Algorithm \ref{alg:Modified-SP}). 
\begin{thm}
[Convergence Speed]\label{Convergence-SpeedSuppose-II}Denote $\rho$
as the total signal energy, i.e., $\rho=\left\Vert \mathbf{X}\right\Vert _{F}^{2}$.
Suppose $\rho>\left(\frac{C_{2}+C_{1}-1}{1-C_{1}}\eta\right)^{2}$
and $\gamma>\frac{C_{2}\eta}{1-C_{1}}$. If $\delta_{s_{2}|d}<0.246$,
then Step 2 of Algorithm \ref{alg:Modified-SP} will stop with no
more than $n_{co}$ iterations where $n_{co}$ is given by

\begin{equation}
n_{co}=\log_{C_{1}}\left[\frac{\gamma-\frac{C_{2}\eta}{1-C_{1}}}{\sqrt{1+\delta_{\bar{s}|d}}\rho^{\frac{1}{2}}+\eta-\frac{C_{2}\eta}{1-C_{1}}}\right].\label{eq:convergence_steps}
\end{equation}
\end{thm}
\begin{proof}
See Appendix \ref{sub:Proof-of-Theorem-con}.
\end{proof}

\begin{remrk}
[Interpretation of Theorem \ref{Convergence-SpeedSuppose-II}]Theorem
\ref{Convergence-SpeedSuppose-II} gives an upper bound on the number
of iterations in the proposed M-SP. Compared with the conventional
SP \cite{dai2009subspace}, our derived convergence result further
cover the cases with measurement noise, which is not discussed by
conventional SP \cite{dai2009subspace} (see Table \ref{tab:Comparison-between-SP-MSP}
for the detailed comparison). On the other hand, from Theorem \ref{Convergence-SpeedSuppose-II},
we obtain that Algorithm \ref{alg:Modified-SP} will converge in $\mathcal{O}\left(\log\mbox{SNR}\right)$
steps in the high SNR (i.e., $\mbox{SNR}\triangleq\frac{\rho}{\eta^{2}}\rightarrow\infty$)
regimes.
\end{remrk}

\section{Robustness to Model Mismatch}

In Section III, we have proposed an M-SP algorithm to exploit the
prior support $\mathcal{T}_{0}$ adaptively based on the support quality
information $s_{c}$. However, in practice, there may be cases with
\emph{incorrect} statistical information $s_{c}$, i.e., $|\mathcal{T}_{0}\bigcap\mathcal{T}|<s_{c}$.
In such scenarios, the proposed M-SP may perform badly. We use the
example below to illustrate this fact. 
\begin{example}
[Algorithm \ref{alg:Modified-SP} with Model Mismatch]\label{Model-MismatchesSuppose}Consider
quality information $s_{c}$ wrongly indicates the quality of the
prior support $\mathcal{T}_{0}$, i.e., $\left|\mathcal{T}_{0}\bigcap\mathcal{T}\right|<s_{c}$,
and $\left|\mathcal{T}\right|=\bar{s}$ (i.e., in $\mathcal{T}$,
only less than $s_{c}$ indices are from $\mathcal{T}_{0}$). With
the proposed M-SP algorithm, from Step 2C, there will always be $s_{c}$
indices selected from $\mathcal{T}_{0}$ while Algorithm \ref{alg:Modified-SP}
will select no more than $\bar{s}-s_{c}$ indices from $\{1,..,K\}\backslash\mathcal{T}_{0}$.
Consequently, the final identified signal support $\hat{\mathcal{T}}$
will always be incorrect. 
\end{example}

From the above example, the performance of the proposed M-SP is sensitive
to model mismatch with incorrect $s_{c}$. In this section, we shall
further propose a \emph{conservative} M-SP approach which will be
robust to scenarios with possible model mismatch. \vspace{-0.5cm}

\begin{center}
\ovalbox{\begin{minipage}[t]{1\columnwidth}%
Challenge 3: Robust algorithm design to combat model mismatch with
incorrect prior support information $s_{c}$.%
\end{minipage}}
\par\end{center}

\subsection{Proposed Conservative M-SP Algorithm}

\begin{table}
\begin{centering}
{\small{}}%
\begin{tabular}{|c|c|}
\hline 
{\small{}$C_{5}$} & {\small{}$\frac{2\delta_{s_{3}|d}\sqrt{1+\delta_{s_{3}|d}}\sqrt{1-\delta_{s_{3}|d}+4\delta_{s_{3}|d}^{2}+4\delta_{s_{3}|d}^{3}}}{\left(1-\delta_{s_{3}|d}\right)^{2}}$}\tabularnewline
\hline 
\multirow{2}{*}{{\small{}$C_{6}$}} & {\small{}$2\sqrt{\frac{1+\delta_{\bar{s}|d}}{1-\delta_{2\bar{s}+s_{c}|d}}}+\sqrt{1+\delta_{\bar{s}|d}}\sqrt{1+\frac{4\delta_{3\bar{s}+s_{c}|d}^{2}(1+\delta_{3\bar{s}+s_{c}|d})}{1-\delta_{2\bar{s}+s_{c}|d}}}$}\tabularnewline
 & {\small{}$\times\left(\frac{2\delta_{3\bar{s}+s_{c}|d}}{\left(1-\delta_{\bar{s}|d}\right)\sqrt{1-\delta_{2\bar{s}|d}}}+\frac{2\sqrt{1+\delta_{\bar{s}|d}}}{1-\delta_{\bar{s}|d}}\right)+1$}\tabularnewline
\hline 
{\small{}$C_{7}$} & {\small{}$\frac{(1-C_{5}+C_{6})}{(1-C_{5})\sqrt{1-\delta_{2\bar{s}|d}}}$}\tabularnewline
\hline 
where & {\small{}$s_{3}\triangleq3\bar{s}+s_{c}+\min\left(0,\left|\mathcal{T}_{0}\right|-\left|\mathcal{T}_{0}\bigcap\mathcal{T}\right|-s_{c}\right)$}\tabularnewline
\hline 
\end{tabular}
\par\end{centering}{\small \par}

\protect\caption{\label{tab:The-detailed-expressions-constants-conservative}The detailed
expressions for constants.}
\end{table}

The conservative M-SP algorithm is obtained by redesigning\emph{ }two
substeps in Step 2 of Algorithm \ref{alg:Modified-SP}. The details
are given in Algorithm \ref{alg:Conservative-Modified-SP}.

\begin{algorithm}
Obtained from Algorithm \ref{alg:Modified-SP} with Step 2A, and Step
2C replaced by the following substeps, respectively:
\begin{itemize}
\item \textbf{Step 2A} (\emph{Support Merge}): Set $\tilde{s}=s_{c}-\left|\hat{\mathcal{T}}_{l}\bigcap\mathcal{T}_{0}\right|$
and merge $\mathcal{T}_{a}=\hat{\mathcal{T}}_{l}\bigcup\mathcal{T}_{b}\bigcup\mathcal{T}_{c}$,
where
\begin{equation}
\mathcal{T}_{b}=\begin{cases}
\arg\max_{|\mathcal{T}_{1}|=\tilde{s},\mathcal{T}_{1}\subseteq\mathcal{T}_{0}}\left\Vert \left(\Phi^{H}\mathbf{R}_{(l)}\right)^{[\mathcal{T}_{1}]}\right\Vert _{F} & \tilde{s}>0\\
\emptyset & \tilde{s}\leq0
\end{cases}.\label{eq:first_region}
\end{equation}
\begin{equation}
\mathcal{T}_{c}=\arg\max_{|\mathcal{T}_{2}|=\bar{s}}\left\Vert \left(\Phi^{H}\mathbf{R}_{(l)}\right)^{[\mathcal{T}_{2}]}\right\Vert _{F}.\label{eq:second_region}
\end{equation}

\item \textbf{Step 2C} (\emph{Support Refinement}): Select $\hat{\mathcal{T}}_{l+1}=\arg\max_{|\mathcal{T}|=\bar{s}}\left\Vert \mathbf{Z}^{[\mathcal{T}]}\right\Vert $
.
\end{itemize}
\protect\caption{\label{alg:Conservative-Modified-SP}Conservative M-SP to Solve Challenge
3.}
\end{algorithm}

\begin{remrk}
[Interpretation of Algorithm \ref{alg:Conservative-Modified-SP}]Note
that in Step 2A of the conservative M-SP, the newly added support
contains two parts, $\mathcal{T}_{b}$ and $\mathcal{T}_{c}$, where
$\mathcal{T}_{b}$ is selected from $\mathcal{T}_{0}$ with size $s_{c}-\left|\hat{\mathcal{T}}_{l}\bigcap\mathcal{T}_{0}\right|$
(compared with $s_{c}$ in M-SP), and $\mathcal{T}_{c}$ is selected
from the entire index space $\{1,..,K\}$ with size $\bar{s}$ (compared
with size $\bar{s}-s_{c}$ selected from $\{1,..,K\}\backslash\mathcal{T}_{0}$
in M-SP). These designs give us opportunities to further search for
support outside $\mathcal{T}_{0}$ when the information of $s_{c}$
is incorrect (i.e., $|\mathcal{T}_{0}\bigcap\mathcal{T}|<s_{c}$).
On the other hand, in Step 2C of the conservative M-SP, the updated
support $\hat{\mathcal{T}}_{l+1}$ with size $\bar{s}$ is selected
from the entire index set $\{1,..,K\}$ (compared with the two part
structure in M-SP). Using this design, even if $s_{c}$ wrongly indicates
the quality of $\mathcal{T}_{0}$, we still have chances to correctly
identify the signal support. Note that the proposed conservative M-SP
still exploits the prior support information but in a conservative
way:\end{remrk}
\begin{itemize}
\item \textbf{Exploitation of Prior Support} $(\mathcal{T}_{0},\, s_{c})$:
For instance, in step 2A, equation (\ref{eq:first_region}) ensures
the selected support candidate $\mathcal{T}_{a}$ contains at least
$s_{c}$ indices from $\mathcal{T}_{0}$. 
\item \textbf{Conservativeness in Exploiting} $(\mathcal{T}_{0},\, s_{c})$:
Compared with the original M-SP, the proposed Algorithm \ref{alg:Conservative-Modified-SP}
exploits $(\mathcal{T}_{0},s_{c})$ in a much more conservative way.
First, in $\mathcal{T}_{a}$ obtained in Step 2, although $\mathcal{T}_{0}$
has already contributed $s_{c}$ indices, another $\bar{s}$ indices
are further selected from the \emph{entire} index space $\{1,..,K\}$
in (\ref{eq:second_region}). Second, the refined support $\hat{\mathcal{T}}_{l+1}$
is obtained from the \emph{entire} index space $\{1,..,K\}$ based
the maximum correlation criterion as in Step 2C (instead of always
selecting $s_{c}$ indices from $\mathcal{T}_{0}$ as in the original
M-SP). These designs allow opportunities to search for the signal
support outside $\mathcal{T}_{0}$. As a result, the proposed conservative
M-SP does not utilize $(\mathcal{T}_{0},\, s_{c})$ wholeheartedly
and hence, is exploiting $(\mathcal{T}_{0},\, s_{c})$ in a more conservatively
way (compared with the M-SP).
\end{itemize}

Recall Example \ref{Model-MismatchesSuppose} with model mismatch
(i.e., $|\mathcal{T}_{0}\bigcap\mathcal{T}|<s_{c}$). Using the conservative
M-SP, both Step 2A and Step 2C would select $\bar{s}$ indices from
the entire index space $\{1,..,K\}$ based on the maximum correlation
criterion \cite{dai2009subspace}. Therefore, the conservative M-SP
has a chance to identify more than $\bar{s}-s_{c}$ indices from $\{1,..,K\}\backslash\mathcal{T}_{0}$
and it is still likely that the correct support $\mathcal{T}$ can
be identified. Hence, the conservative M-SP is robust to model mismatch
with incorrect $s_{c}$. We formally discuss this fact in the next
Section.

\subsection{Performance Analysis of Conservative M-SP}

In this Section, we shall analyze the recovery performance of the
proposed conservative M-SP. Specifically, we give similar results
as in Section IV except that the the derived results in this section
\emph{do not require} the assumption that the information $s_{c}$
is correct. 
\begin{thm}
[Distortion Bound of Conservative M-SP]\label{Distortion-Bound-of-conservative_MSP}Suppose
the $s_{3}$-th order block-RIP constant $\delta_{s_{3}|d}$ satisfies
$\delta_{s_{3}|d}<0.246$. We obtain the following results regarding
Algorithm \ref{alg:Conservative-Modified-SP}:

(i) The obtained solution $\mathbf{\hat{X}}$ satisfies
\begin{equation}
\left\Vert \mathbf{X}-\hat{\mathbf{X}}\right\Vert _{F}\leq\max\left(C_{7}\eta,\quad\frac{\gamma+\eta}{\sqrt{1-\delta_{2\bar{s}|d}}}\right)\label{eq:dist_bound-conservative}
\end{equation}

(ii) If $\mathbf{X}$ satisfies {\small{}$\min_{k\in\mathcal{T}}\left\Vert \mathbf{X}[k]\right\Vert _{F}>\max\left(C_{7}\eta,\quad\frac{\gamma+\eta}{\sqrt{1-\delta_{s_{1}|d}}}\right)$,
}then $\mathbf{\hat{X}}$ further satisfies
\begin{equation}
\left\Vert \mathbf{X}-\hat{\mathbf{X}}\right\Vert _{F}\leq\frac{1}{\sqrt{1-\delta_{\bar{s}|d}}}\eta\label{eq:refined_bound-conservative}
\end{equation}
where $s_{3}$, $C_{5}$, $C_{6}$, $C_{7}$ depends on the block-RIP
constants and are given in Table \ref{tab:The-detailed-expressions-constants-conservative}. \end{thm}
\begin{proof}
See Appendix \ref{sub:Proof-of-Theorem-conservative_MSP}.
\end{proof}

\begin{thm}
[Convergence Speed of Conservative M-SP]\label{Convergence-Speed-MSP}Denote
$\rho$ as the signal energy, i.e, $\rho=\left\Vert \mathbf{X}\right\Vert _{F}^{2}$.
Suppose $\rho>\left(\frac{C_{5}+C_{6}-1}{1-C_{5}}\eta\right)^{2}$
and $\gamma>\frac{C_{6}\eta}{1-C_{5}}$. If $\delta_{s_{3}|d}<0.246$,
then in Algorithm \ref{alg:Conservative-Modified-SP}, Step 2 will
stop with no more than $n_{co}$ iterations where $n_{co}$ is given
by

\begin{equation}
n_{co}=\log_{C_{5}}\left[\frac{\gamma-\frac{C_{6}}{1-C_{5}}\eta}{\sqrt{1+\delta_{\bar{s}|d}}\rho^{\frac{1}{2}}+\eta-\frac{C_{6}}{1-C_{5}}\eta}\right].\label{eq:convergence_steps-1}
\end{equation}
\end{thm}
\begin{proof}
(Sketch) The proof is similar to Appendix \ref{Convergence-SpeedSuppose-II}
and is therefore omitted to avoid duplication.
\end{proof}

\begin{remrk}
[Interpretation of Theorem \ref{Distortion-Bound-of-conservative_MSP}-\ref{Convergence-Speed-MSP}]Different
from the theoretical results derived for M-SP in Section IV, Theorem
\ref{Distortion-Bound-of-conservative_MSP}-\ref{Convergence-Speed-MSP}
for the proposed conservative M-SP (Algorithm \ref{alg:Conservative-Modified-SP})\emph{
do not depend on the assumption} of correct quality information $s_{c}$
(i.e., no matter whether $\left|\mathcal{T}_{0}\bigcap\mathcal{T}\right|\geq s_{c}$
is true or not). These results demonstrate the \emph{robustness} of
the proposed conservative M-SP towards model mismatch with incorrect
$s_{c}$. Note that compared with the M-SP, there is an increase on
the requirement of the block-RIP conditions as can be seen from the
expression of $s_{3}$ in $\delta_{s_{3}|d}$ in Table \ref{tab:The-detailed-expressions-constants-conservative}
(i.e., $s_{3}\geq s_{2}$). This is due to the \emph{conservative}
exploitation of $(\mathcal{T},\, s_{c})$ in Algorithm \ref{alg:Conservative-Modified-SP}
such that in Step A, a larger support candidate is involved in the
signal support identification. 
\end{remrk}

\section{Application to Sparse Channel Estimation in Massive MIMO}

In this section, we shall apply the proposed framework of CS to the
channel estimation problem in massive MIMO \cite{larsson2013massive}
with temporal correlation. One key challenge to implement massive
MIMO is to efficiently obtain the channel state information at the
transmitter (CSIT). Recently, it has been shown that the massive MIMO
channel is sparse due to the limited local scatterers effect \cite{zhou2006experimental,kyritsi2003correlation}
and hence CS techniques are deployed to reduce the CSI acquisition
overhead by exploiting the channel sparsity. For instance, in \cite{kuo2012compressive},
CS techniques are deployed to improve the channel feedback efficiency
and in \cite{rao2014distributed}, a distributed CS framework is proposed
to enhance both the channel estimation and feedback performance in
downlink massive MIMO systems. Besides, works \cite{nguyen2013compressive}
and \cite{wen2014channel} further consider uplink massive MIMO systems,
and a CS-based low-rank approximation scheme and a sparse Bayesian-learning
algorithm respectively, are proposed to improve the channel recovery
performance. However, these existing approaches \cite{zhou2006experimental,kyritsi2003correlation}
only consider a one-time slot static scenario. In massive MIMO systems
with temporarily correlated multipaths (as illustrated in Figure \ref{fig:Illustration-of-point-to-point}),
it is desirable to exploit the channel temporal correlation to further
reduce the required pilot overhead. In this section, we share achieve
this goal by applying the proposed framework of CS recovery with prior
support information.

\subsection{System Model}

Consider a flat block-fading FDD massive MIMO system with one BS and
one UE, where the BS and UE have $M$ ($M$ is large) and $N$ antennas
respectively. To estimate the downlink channel from the BS to the
UE, the BS sends a sequence of $T$ training pilot symbols on its
$M$ antennas. Denote the transmitted pilot training matrix as $\Theta\in\mathbb{C}^{M\times T}$
where $\textrm{tr}(\Theta\Theta^{H})=T$. The corresponding received
signal at the UE $\mathbf{Z}\in\mathbb{C}^{N\times T}$ is
\begin{equation}
\mathbf{Z}=\sqrt{P}\mathbf{H}\Theta+\mathbf{W}\label{eq:CSIT_model}
\end{equation}
where $P$ denotes the transmitted SNR from the BS, $\mathbf{H}\in\mathbb{C}^{N\times M}$
is the quasi-static channel from the BS to the UE, $\mathbf{W}\in\mathbb{C}^{N\times T}$
is the channel noise whose elements are i.i.d. complex Gaussian variables
with zero mean and unit variance. Our \emph{target} is to estimate
the channel matrix $\mathbf{H}$ based on the obtained channel observations
$\mathbf{Z}$ at the UE. We first elaborate the considered channel
model in the next subsection.

\subsection{Channel Model}

Consider a uniform linear array (ULA) model for the antennas installed
at the BS and UE. The channel matrix $\mathbf{H}$ can be represented
\cite{tse2005fundamentals} as
\[
\mathbf{H}=\mathbf{U}\mathbf{H}_{a}\mathbf{V}^{H}
\]
where $\mathbf{U}\in\mathbb{C}^{N\times N}$ and $\mathbf{V}\in\mathbb{C}^{M\times M}$
denote the unitary matrices for the angular domain transformation
at the UE and BS side respectively, $\mathbf{H}_{a}\in\mathbb{C}^{N\times M}$
is the angular domain channel matrix. In massive MIMO systems, due
to the limited local scattering at the BS side, the angular domain
channel $\mathbf{H}_{a}$ turns out to be sparse. Furthermore, as
the UE has a relatively rich number of local scatterers compared with
its number of antennas, the angular domain $\mathbf{H}_{a}$ has simultaneous
zero or non-zero columns, as indicated in \cite{kyritsi2003correlation,rao2014distributed}
(illustrated in Figure \ref{fig:Illustration-of-point-to-point}).
Figure \ref{fig:Illustration-of-angular} illustrates the simulated
results of the angular domain channel using the ITU-R IMT-Advanced
channel model \cite{3GPPchannel}. Based on these features and similar
to \cite{kyritsi2003correlation,rao2014distributed}, we consider
the following channel model for our point-to-point massive MIMO system.
Denote $\textrm{supp}(\mathbf{h})=\{i:\mathbf{h}(i)\neq0\}$. 

\begin{figure}
\begin{centering}
\includegraphics[width=3.5in]{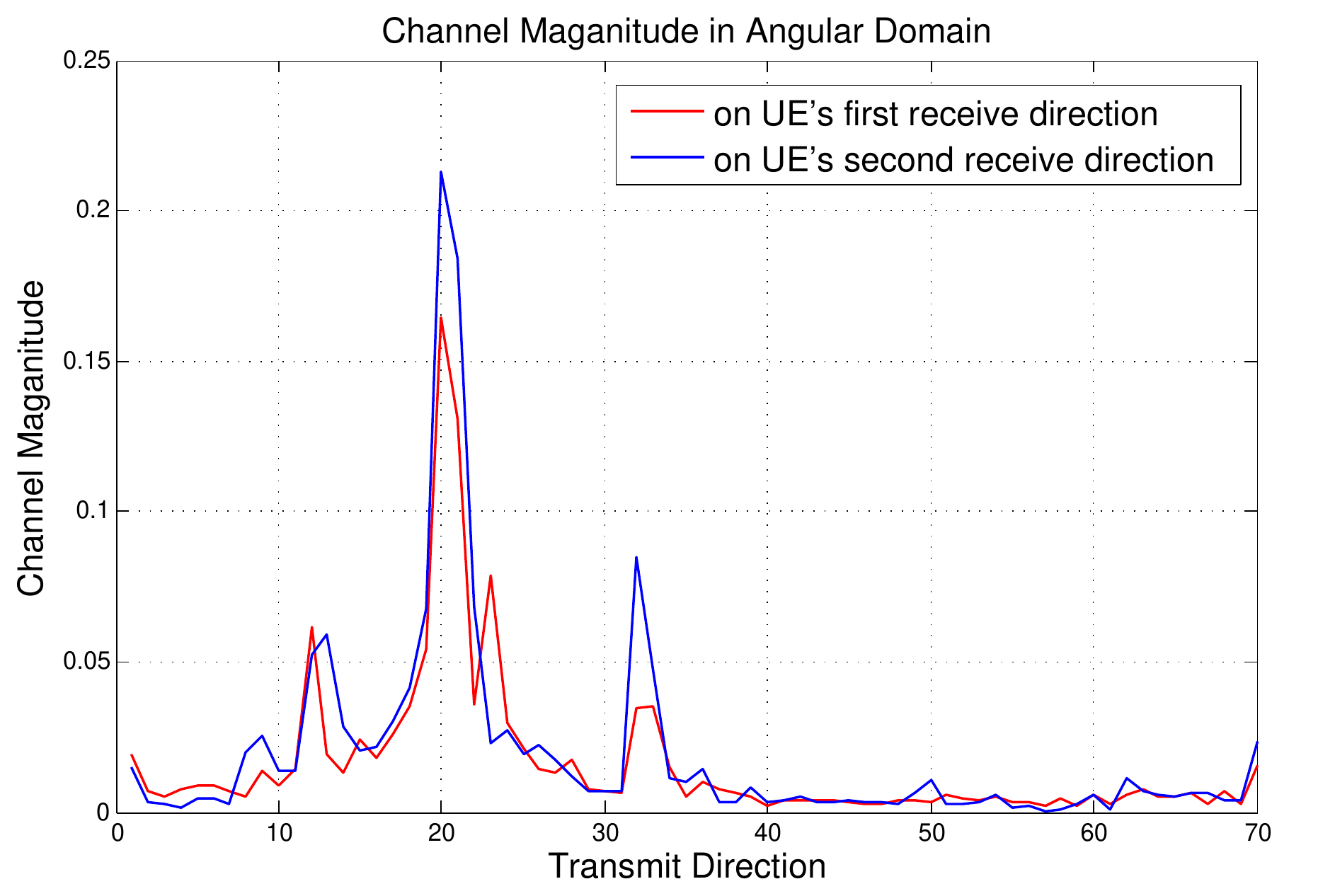}
\par\end{centering}

\protect\caption{\label{fig:Illustration-of-angular}Illustration of angular domain
channel with the ITU-R IMT-Advanced model in Urban Micro scenario
\cite{3GPPchannel}. The number of antennas at the BS and UE are 70
and 2, respectively. As can been seen, (i) the angular domain channel
are sparse; (ii) the angular domain channel on different receive directions
has simultaneous channel support. }

\end{figure}

\begin{definitn}
[Massive MIMO Channel Model]Let $\mathbf{h}_{j}\in\mathbb{C}^{N\times1}$
be the $j$-th row vector of $\mathbf{H}_{a}\in\mathbb{C}^{N\times M}$.
The channel matrix $\mathbf{H}_{a}$ satisfies: $\textrm{supp}(\mathbf{h}_{1})=\cdots\textrm{supp}(\mathbf{h}_{N})\triangleq\mathcal{T}$,
where $\mathcal{T}$ is the channel support and $\left|\mathcal{T}\right|\leq\bar{s}$.
Furthermore, the elements in $(\mathbf{H}_{a})_{\mathcal{T}}$ are
i.i.d. complex Gaussian variables with zero mean and unit variance. 
\end{definitn}

Note that $\bar{s}$ is a statistical upper bound on the number of
spatial paths from the BS to the UE. In practice, the channel sparsity
levels depend on the large scale properties of the scattering environment
and changes slowly and hence information like $\bar{s}$ can be obtained
at the UE from prior offline measurements. On the other hand, the
channel paths are \emph{temporarily correlated} so that consecutive
frames would share some common channel paths. As a result, we can
utilize the \emph{prior support informatio}n (Definition \ref{Temporarily-Correlated-Sparse})
in massive MIMO scenarios. Specifically, in the prior support information
$(\mathcal{T}_{0},s_{c})$, $\mathcal{T}_{0}$ is the estimated channel
support in the \emph{previous} frame and $s_{c}$ characterizes the
size of common channel paths between $\mathcal{T}_{0}$ and $\mathcal{T}$,
i.e., $\left|\mathcal{T}_{0}\bigcap\mathcal{T}\right|\geq s_{c}$. 
\begin{remrk}
[Practical Considerations]\label{Practical-Considerations}In practice,
we usually need to estimate a sequence of channels $\mathbf{H}_{[1]}$,
$\mathbf{H}_{[2]}$ ... where $\mathbf{H}_{[i]}$ is the channel from
the BS to the UE in the $i$-th frame \cite{3GPPchannel}. At the
very beginning, we don't have prior channel estimations and hence
we can set the prior channel support information $(\mathcal{T}_{0},s_{c})$
to be $\mathcal{T}_{0}=\emptyset$, $s_{c}=0$. At later stages when
we have already obtained some prior channel estimations, the estimated
channel support in the previous frame (e.g., $\mathbf{H}_{[i-1]}$)
can act as the prior support $\mathcal{T}_{0}$ for the present time
(e.g., $\mathbf{H}_{[i]}$). On the other hand, due to the slowly
varying propagation environment between the BS and UE \cite{huang2009limited,bajwa2010compressed}
(as illustrated in Figure \ref{fig:Illustration-of-point-to-point}),
it is likely that the size of the common support between consecutive
channels, i.e., $\mathbf{H}_{[i-1]}$, $\mathbf{H}_{[i]}$, changes
slowly so that we can gradually obtain a \emph{reliable} statistical
information as $s_{c}$. For instance, we can select $s_{c}$ to satisfy
$\textrm{Pr}(\left|\mathcal{T}_{i-1}\bigcap\mathcal{T}_{i}\right|\geq s_{c})\geq1-\epsilon$
for some small $\epsilon$, $0<\epsilon<1$ from prior channel measurements
based on long term stochastic learning and estimation \cite{bottou2002stochastic}.
Note that a larger $s_{c}$ indicates a stronger temporal correlation
between channels of consecutive frames. \hfill \QED
\end{remrk}

\begin{figure}
\begin{centering}
\includegraphics[scale=0.6]{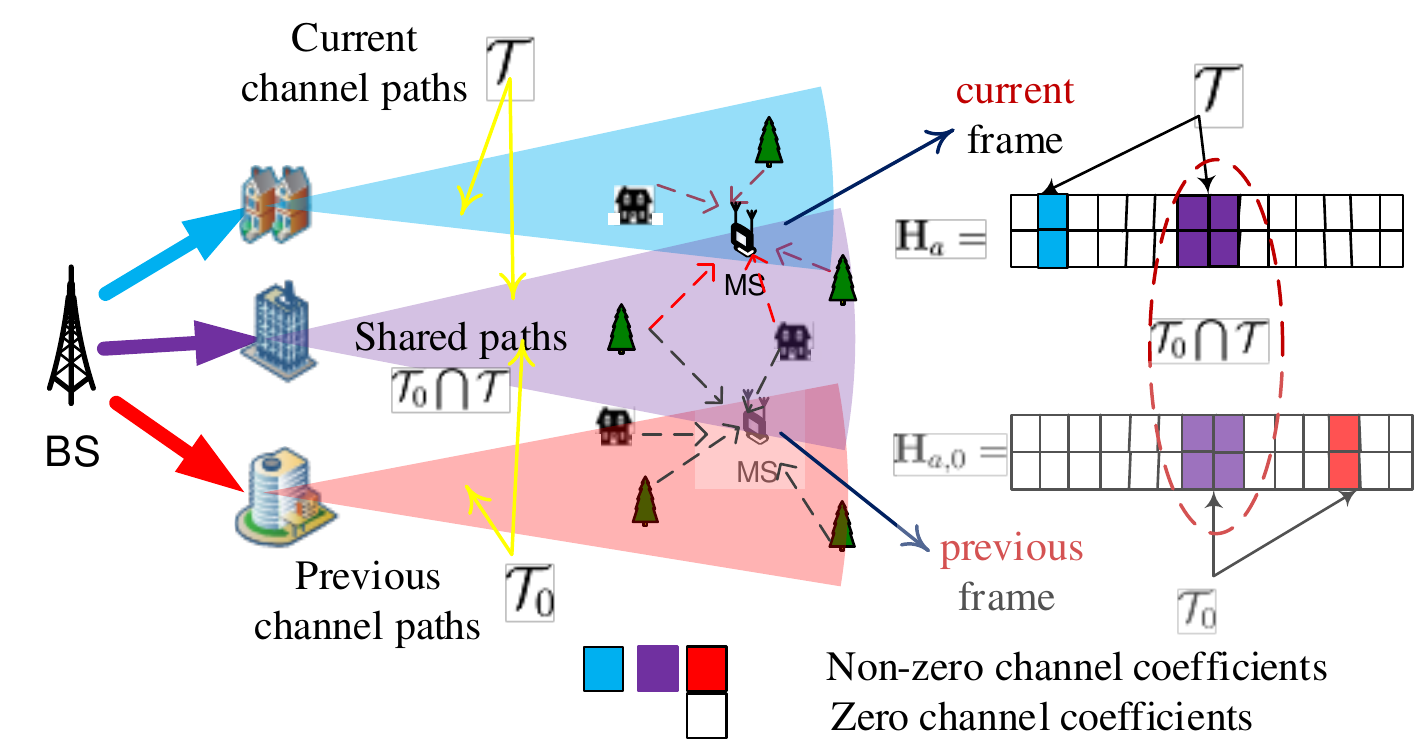}
\par\end{centering}

\protect\caption{\label{fig:Illustration-of-point-to-point}Illustration of point-to-point
massive MIMO system in which the previous and current frames share
some common spatial channel paths $\mathcal{T}_{0}\bigcap\mathcal{T}$
due to the slowly varying scattering environment. As such, the estimated
channel support $\mathcal{T}_{0}$ in the previous frame can be utilized
to enhance the CSIT estimation performance in the current frame. }
\end{figure}

\subsection{Channel Recovery with the Proposed CS Framework}

In this subsection, we talk about how to apply the proposed CS framework
to conduct the recovery of $\mathbf{H}$ based on $\mathbf{Y}$. 

\vspace{-0.5cm}

\begin{center}
\ovalbox{\begin{minipage}[t]{1\columnwidth}%
Challenge 4: Apply the proposed framework of CS with prior support
information in Section II, to conduct the recovery of $\mathbf{H}$
from (\ref{eq:CSIT_model}).%
\end{minipage}}
\par\end{center}

First, equation (\ref{eq:CSIT_model}) can be re-written as 
\begin{equation}
\underset{\mathbf{Y}}{\underbrace{(\mathbf{\mathbf{Z}}^{H}\mathbf{U})}}=\underset{\Phi}{\underbrace{\sqrt{\frac{M}{T}}(\mathbf{V}\Theta)^{H}}}\underset{\mathbf{X}}{\underbrace{\sqrt{\frac{PT}{M}}(\mathbf{H}_{a})^{H}}}+\underset{\mathbf{N}}{\underbrace{\mathbf{W}^{H}\mathbf{U}}}.\label{eq:revised_CSIT_model}
\end{equation}
Then (\ref{eq:revised_CSIT_model}) matches the CS measurement model
in (\ref{eq:CS_signal_model}), where $(\mathbf{\mathbf{Z}}^{H}\mathbf{U})$
are measurements (role of $\mathbf{Y}$ in (\ref{eq:CS_signal_model})),
$\sqrt{\frac{M}{T}}(\mathbf{V}\Theta)^{H}$ is the measurement%
\footnote{Note that the term $\sqrt{\frac{M}{T}}$is to normalize the measurement
matrix $\Phi=\sqrt{\frac{M}{T}}(\mathbf{V}\Theta)^{H}$ to satisfy
$\textrm{tr}\left(\Phi\Phi^{H}\right)=M$ so as to fit into the analytical
framework of block-RIP property in Definition \ref{Block-Restricted-Isometry}. %
} matrix ($\Phi$ in (\ref{eq:CS_signal_model})), $\mathbf{W}^{H}\mathbf{U}$
is the noise ($\mathbf{N}$ in (\ref{eq:CS_signal_model})) and $\sqrt{\frac{PT}{M}}(\mathbf{H}_{a})^{H}$
is the unknown signal source ($\mathbf{X}$ in (\ref{eq:CS_signal_model})).
Furthermore, $\sqrt{\frac{PT}{M}}(\mathbf{H}_{a})^{H}$ satisfies
the general sparsity model in Section II-B with chunk size $1\times N$
($d=1$, $L=N$ as in Definition \ref{Matrix-with-Chunk}). As such,
the channel recovery problem is transformed the CS problem we consider
in Section II. 

Second, based on the transformed CS equation (\ref{eq:revised_CSIT_model}),
we apply the proposed M-SP (Algorithm \ref{alg:Modified-SP}) to conduct
the channel recovery by replacing the \emph{input }parameter $\mathbf{Y}$
with $(\mathbf{\mathbf{Z}}^{H}\mathbf{U})$, $\Phi$ with $(\mathbf{V}\Theta)^{H}$,
with $d$ set to be $d=1$. Denote the obtained algorithm \emph{output
}as $\hat{\mathbf{X}}$. Then the recovered channel $\hat{\mathbf{H}}$
for $\mathbf{H}$ is given by
\begin{equation}
\hat{\mathbf{H}}=\sqrt{\frac{M}{PT}}\mathbf{U}(\hat{\mathbf{X}})^{H}\mathbf{V}^{H}.\label{eq:recovered_channel}
\end{equation}

Third, we deploy the analytical results in Section IV to derive some
performance results for $\hat{\mathbf{H}}$. Note that when $d=1$,
the block-RIP is reduced to the conventional RIP \cite{candes2005decoding}.
Suppose that the pilot matrix $\sqrt{\frac{M}{T}}\Theta^{H}$ satisfies
the RIP property and denote the corresponding $k$-th order RIP constants
as $\delta_{k}$ (note that $\delta_{k}=\delta_{k|1}$ as $d=1$).
Based on Theorem \ref{Distortion-BoundDenote-MSP} and from the unitary
invariance property of Frobenius norm, we obtain the following distortion
bound.
\begin{thm}
[Channel Recovery Performance]\label{Channel-Recovery-Performance}If
the $s_{2}$-th order RIP constant of $\Phi=\sqrt{\frac{M}{T}}\Theta^{H}$
satisfies $\delta_{s_{2}}\leq0.246$, where $s_{2}=3\bar{s}+\min\left(0,\,\left|\mathcal{T}_{0}\right|-3s_{c}\right)$,
then the average channel recovery distortion, i.e., $\mathbb{E}\left(\left\Vert \hat{\mathbf{H}}-\mathbf{H}\right\Vert _{F}\right)$
satisfies{\small{}
\begin{align}
 & \mathbb{E}\left(\left\Vert \hat{\mathbf{H}}-\mathbf{H}\right\Vert _{F}\right)\nonumber \\
\leq & \sqrt{\frac{M}{PT}}\left(\left(C_{4}+\frac{1}{\sqrt{1-\delta_{s_{2}}}}\right)\frac{\Gamma\left(NT+\frac{1}{2}\right)}{\Gamma(NT)}+\frac{\gamma}{\sqrt{1-\delta_{s_{2}}}}\right)\label{eq:equation_result}
\end{align}
}where $\gamma$ is threshold parameter in Algorithm \ref{alg:Modified-SP},
$\Gamma(\cdot)$ is the gamma function, $C_{4}$ is a constant given
in Table \ref{tab:The-detailed-expressions-constants}.\end{thm}
\begin{proof}
From Theorem \ref{Distortion-BoundDenote-MSP}, equation (\ref{eq:recovered_channel}),
$\delta_{s_{2}}\leq0.246$ and $s_{1}\leq s_{2}$, we derive $\left\Vert \hat{\mathbf{H}}-\mathbf{H}\right\Vert _{F}\leq\sqrt{\frac{M}{PT}}\left(C_{4}\eta+\frac{\gamma+\eta}{\sqrt{1-\delta_{s_{2}}}}\right)$.
From this and $\mathbb{E}\left(\eta\right)=\mathbb{E}\left(\left\Vert \mathbf{W}\right\Vert _{F}\right)=\frac{\Gamma\left(NT+\frac{1}{2}\right)}{\Gamma(NT)}$,
equation (\ref{eq:equation_result}) is derived.
\end{proof}

From Theorem \ref{Channel-Recovery-Performance}, as the transmit
SNR $P\rightarrow\infty$, the average recovery distortion $\mathbb{E}\left(\left\Vert \hat{\mathbf{H}}-\mathbf{H}\right\Vert _{F}\right)\rightarrow0$
and perfect channel recovery will be achieved. On the other hand,
from the expression of $s_{2}\triangleq3\bar{s}+\min\left(0,\,\left|\mathcal{T}_{0}\right|-3s_{c}\right)$,
$s_{2}$ decreases as $s_{c}$ increases when $s_{c}\geq\frac{1}{3}\left|\mathcal{T}_{0}\right|$.
In other words, a weaker RIP condition on the measurement matrix $\sqrt{\frac{M}{T}}\Theta^{H}$
is required as $s_{c}$ increases (e.g., we need $\delta_{3\bar{s}}\leq0.246$
for $s_{c}=0$ and $\delta_{\bar{s}}\leq0.246$ for $s_{c}=\left|\mathcal{T}_{0}\right|=\bar{s}$).
This leads to a smaller requirement on the number of training pilot
$T$ \cite{candes2005decoding}. From this, we conclude that a larger
strength of temporal correlation on the channel support (i.e., larger
$s_{c}$) can enjoy a better reduction on the number of training pilots
in massive MIMO systems. On the other hand, if we apply the conservative
M-SP (Algorithm \ref{alg:Conservative-Modified-SP}) instead of M-SP
(Algorithm \ref{alg:Modified-SP}) to conduct the channel recovery,
we can obtained a similar recovery performance result as in Theorem
\ref{Channel-Recovery-Performance} by deploying Theorem \ref{Distortion-Bound-of-conservative_MSP}
(details are omitted to avoid duplication).

\subsection{Discussion on the Pilot Matrix $\Theta$}

Note that we have not discussed the design of the pilot matrix $\Theta$
so that the aggregate measurement matrix $\Phi=\sqrt{\frac{M}{T}}\Theta^{H}$
in (\ref{eq:revised_CSIT_model}) can satisfy the RIP condition in
Theorem \ref{Channel-Recovery-Performance}. In the CS literature,
matrices randomly generated from sub-Gaussian distribution \cite{foucart2013mathematical}
can satisfy the RIP with overwhelming probability and this randomized
generation method has also been widely used. Following this convention,
the elements of the pilot matrix $\Theta\in\mathbb{C}^{M\times T}$
can be generated from i.i.d. sub-Gaussian distribution (e.g., $\{\sqrt{\frac{1}{M}},-\sqrt{\frac{1}{M}}\}$
with equal probability). Using this method, from \cite{candes2005decoding},
when the length $T$ of the training pilot satisfies $T\geq c_{1}k\log M$,
the probability that the CS measurement matrix $\Phi=\sqrt{\frac{M}{T}}\Theta^{H}$
in (\ref{eq:revised_CSIT_model}) satisfies a prescribed $k$-th order
RIP condition $\delta_{k}\triangleq\delta$ will be no less than $1-\mathcal{O}\left(\textrm{exp}(-c_{2}T)\right)$,
where $c_{1}$ and $c_{2}$ are some positive constants depending
on $\delta$ \cite{candes2005decoding}.

\subsection{Discussion of Other Possible Applications}

In fact, the proposed framework can potentially be applied to many
other areas, including wireless sensor networks (WSN) \cite{xie2014transmission}
and magnetic resonance imaging (MRI) \cite{vasanawala2009faster}
 in which the target sparse signals usually demonstrate strong temporal
correlations. To apply the proposed scheme, one can learn the statistical
information $s_{c}$ (which characterizes the size of the shared common
support between two consecutive signals) using the tools of stochastic
learning and estimation \cite{bottou2002stochastic}. In this work,
we have proposed two algorithms, namely the M-SP and the conservative
M-SP to conduct the signal recovery. For a specific application scenario,
if the uncertainty on $s_{c}$ is small%
\footnote{e.g., when the size of the common support between consecutive signals
changes slowly, one can learn \cite{bottou2002stochastic} a reliable
statistic information $s_{c}$ such that $\textrm{Pr}(\left|\mathcal{T}_{0}\bigcap\mathcal{T}\right|\geq s_{c})\geq1-\epsilon$
for some small $\epsilon$, $0<\epsilon<1$.%
}, then one should use the M-SP algorithm for better performance. On
the other hand, if the underlying uncertainty on model parameter of
$s_{c}$ is large, then one would prefer conservative M-SP for robustness.
The robustness of conservative M-SP with respect to model mismatch
on $s_{c}$ is illustrated in Figure 7 (will be elaborated in Section
VII.D).

\section{Numerical Results}

In this Section, we consider the scenario of sparse channel estimation
in massive MIMO systems as in Section VI to verify the effectiveness
of the proposed framework. Specifically, we compare the performance
of the proposed M-SP and conservative M-SP with the following baselines:
\begin{itemize}
\item \textbf{Baseline 1} (\emph{SP}): Deploy conventional SP \cite{dai2009subspace}
to recover the massive MIMO channel.
\item \textbf{Baseline 2} (\emph{Basis Pursuit}): Deploy conventional basis
pursuit \cite{dai2009subspace} to recover the channel.
\item \textbf{Baseline 3} (\emph{modified Basis Pursuit}): Deploy the modified
basis pursuit proposed in \cite{vaswani2010modified} to recover the
channel with blind exploitation of the prior support information.
\item \textbf{Baseline 4} (\emph{MMV-SP}): Deploy an improved version of
the SP \cite{dai2009subspace} (corresponds the proposed M-SP with
$s_{c}=0$) to adapt to the general sparsity model but without exploitation
of the prior support information.
\item \textbf{Baseline 5} \emph{(AMP-MMV)}: Deploy the approximate message
passing for multiple measurement vector problems (AMP-MMV) to conduct
the channel recovery \cite{ziniel2013efficient}.
\item \textbf{Baseline 6} (Genie-aided LS): This serves as a performance
upper bound scenario, in which the channel support $\mathcal{T}$
is assumed to be known and we directly use least square to recover
the channel coefficients on $\mathcal{T}$.
\end{itemize}

We consider a narrow band (flat fading) point-to-point massive MIMO
system with one BS and one UE, where the BS and UE have $M=200$ and
$N=2$ antennas, respectively. Denote the average transmit SNR at
the BS as $P$. We use the 3GPP spatial channel model (SCM) \cite{3GPPchannel}
to generate the channel coefficients and we consider that the UE has
a rich local scattering environment as in \cite{klessling2003mimo}.
Denote the channel to be estimated in the $i$-th frame as $\mathbf{H}_{i}$
and denote its corresponding channel support as $\mathcal{T}_{i}$.
Suppose that the number of spatial paths from the BS broadside (corresponding
to $|\mathcal{T}_{i}|$) are randomly generated as $|\mathcal{T}_{i}|\sim\mathcal{U}\left(\bar{s}-2,\bar{s}\right)$,
$\forall i$, where $\mathcal{U}(a,b)$ denotes discrete uniform distribution
over the set of integers $\{a,a+1,...,b\}$. Consider a slowly varying
scattering scenario so that consecutive frames (i.e., $\mathcal{T}_{i}$,
$\mathcal{T}_{i+1}$) share some spatial channel paths with size $|\mathcal{T}_{i}\bigcap\mathcal{T}_{i+1}|\sim\mathcal{U}\left(s_{c},s_{c}+2\right)$.
The threshold parameter $\gamma$ in the proposed M-SP and conservative
M-SP are given by $\gamma=\sqrt{2NT}$, where $T$ is the length of
the training pilots. In baseline 2 \cite{candes2005decoding} and
baseline 3 \cite{vaswani2010modified}, the threshold parameters in
the constraint of the $l_{1}$-norm minimization are also set to be
$\sqrt{2NT}$. In the following, we compare the normalized mean squared
error (NMSE)%
\footnote{The NMSE of the estimated channel is computed as $\frac{1}{G}\sum_{i=1}^{G}\frac{\left\Vert \mathbf{H}_{i}-\mathbf{\hat{H}}_{i}\right\Vert _{F}^{2}}{\left\Vert \mathbf{H}_{i}\right\Vert _{F}^{2}}$,
where $\mathbf{H}_{i}$ and $\mathbf{\hat{H}}_{i}$ are the actual
channel and the estimated channel, in the $i$-th realization respectively,
$G$ is the number of simulation realizations.%
} \cite{sideratos2007advanced} of the estimated channel with $G=1000$
channel realizations.

\subsection{Channel Estimation Performance Versus Overhead $T$}

In Figure \ref{fig:NMSE-versus_T}, we compare the normalized mean
squared error (NMSE) \cite{sideratos2007advanced} of the estimated
channel versus the length of the training pilot $T$, under transmit
SNR $P=25$ dB, channel sparsity parameter $\bar{s}=18$, and prior
channel quality parameter $s_{c}=10$. From this figure, we observe
that the channel estimation performance increases as $T$ increases,
and the proposed M-SP algorithm achieves a substantial performance
gain over the Baseline 1-4. This is because the proposed M-SP adaptively
exploits the prior channel support based on its quality parameter
$s_{c}$ and it also adapts to the joint channel sparsity structure
as illustrated in Section VI. Specifically, the performance gain of
the M-SP over MMV-SP demonstrates the advantage of adaptively exploiting
the prior channel support, and the performance gain of MMV-SP over
SP indicates the benefits of adapting to the joint sparsity structures.
On the other hand, note that the proposed conservative MSP has a \emph{smaller}
performance gain compared with the proposed M-SP. This is because
the conservative M-SP utilizes the prior channel support in a more
conservative manner and hence achieves less exploitation gain.

\begin{figure}
\centering{}\includegraphics[width=3.5in]{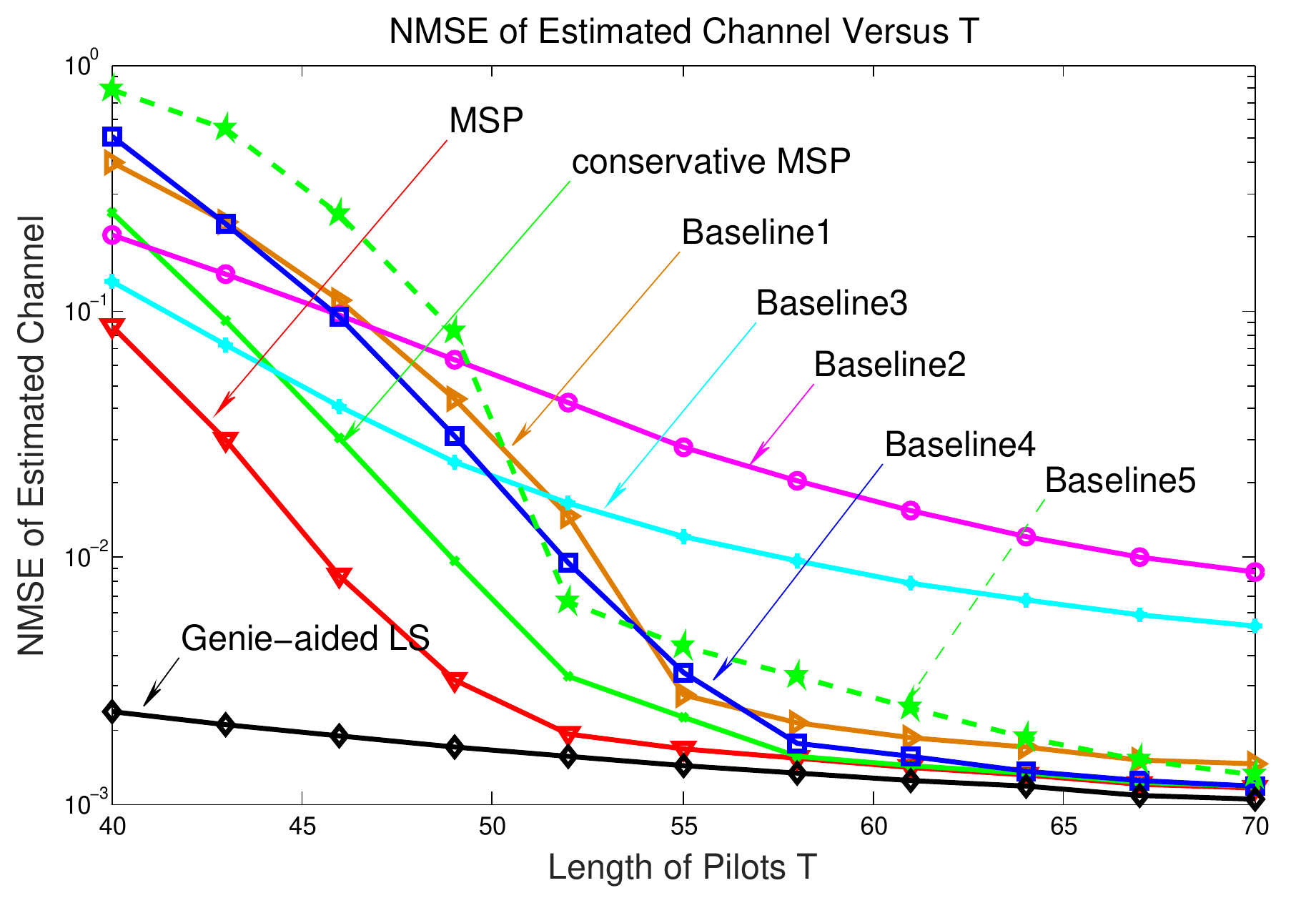}\protect\caption{\label{fig:NMSE-versus_T}NMSE of estimated channel versus the pilot
training length $T$ under $\bar{s}=18$, $s_{c}=10$ and transmit
SNR $P=25$ dB.}
\end{figure}

\begin{figure}
\begin{centering}
\includegraphics[width=3.5in]{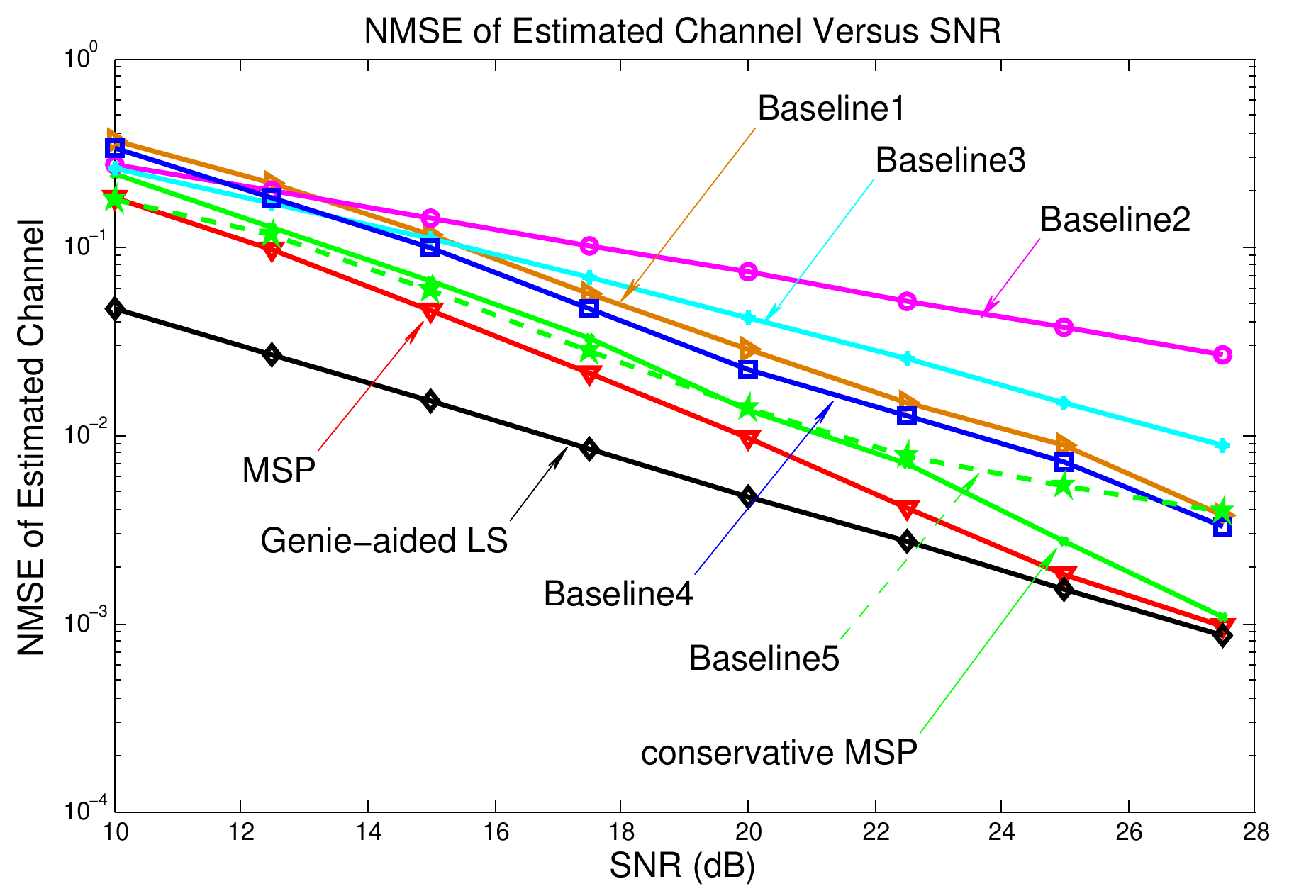}
\par\end{centering}

\protect\caption{\label{fig:NMSE-versus_SNR}NMSE of estimated channel versus transmit
SNR under $T=52$ and $\bar{s}=18$, $s_{c}=10$.}
\end{figure}

\begin{figure}
\begin{centering}
\includegraphics[width=3.5in]{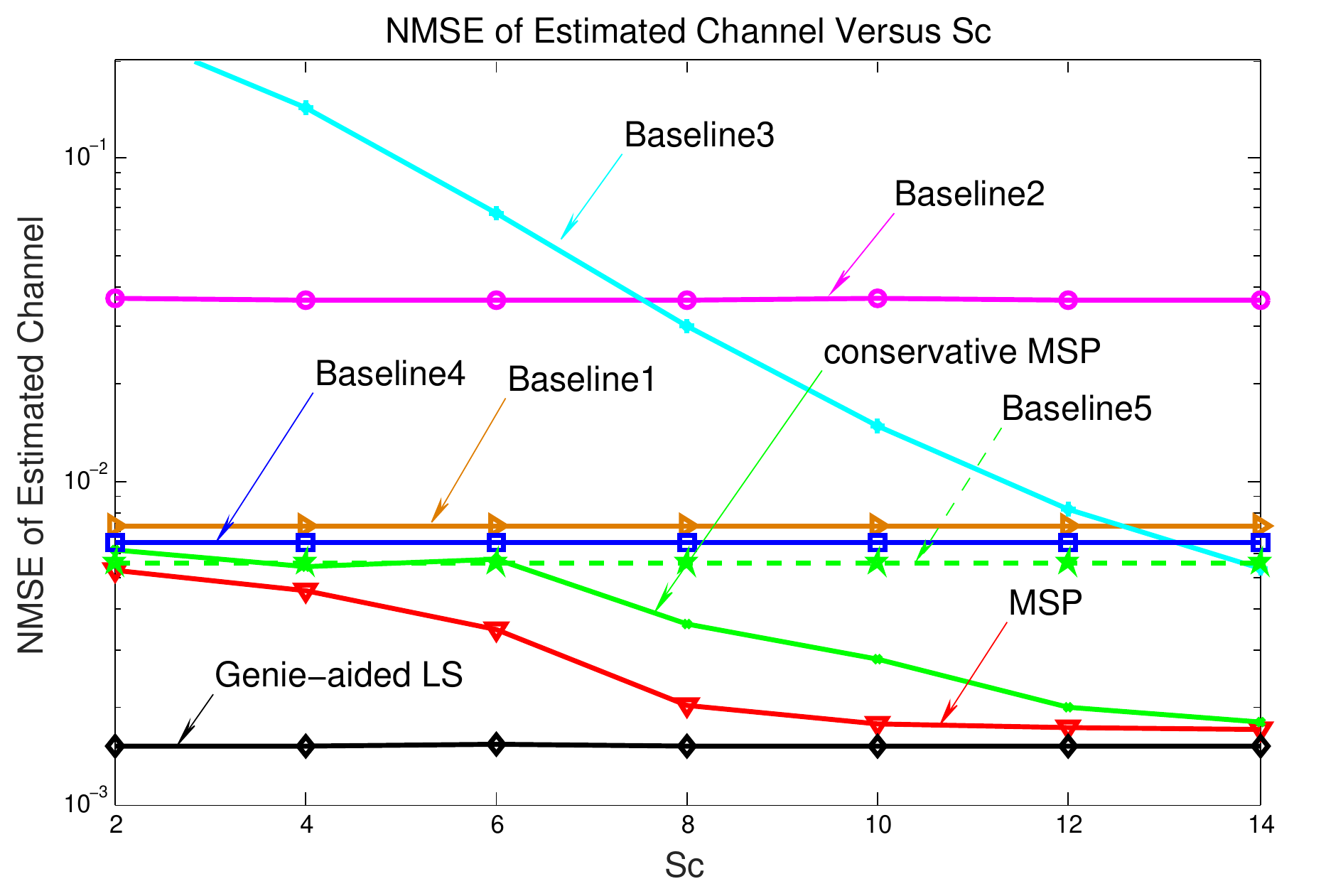}
\par\end{centering}

\protect\caption{\label{fig:NMSE-versus_Sc}NMSE of estimated channel versus prior
channel support quality parameter $s_{c}$ under transmit SNR $P=$25
dB, $T=52$ and $\bar{s}=18$.}
\end{figure}

\subsection{Channel Estimation Performance Versus Transmit SNR $P$}

In Figure \ref{fig:NMSE-versus_SNR}, we compare the NMSE of the estimated
channel versus the transmit SNR $P$ under $T=52$, $\bar{s}=18$
and $s_{c}=10$. From this figure, we observe that the proposed M-SP
algorithm has substantial performance gain over the baselines and
relatively a larger performance gain is achieved in higher SNR regions.

\subsection{Channel Estimation Performance Versus Temporal Correlation Strength
$s_{c}$}

In Figure \ref{fig:NMSE-versus_Sc}, we compare the NMSE of the estimated
channel versus the prior support quality parameter $s_{c}$ (which
indicates the strength of temporal correlation between channels of
consecutive frames) under $T=52$, $\bar{s}=18$ and $P=25$ dB. From
this figure, we observe that the channel estimation performance of
the proposed M-SP and conservative M-SP gets better $s_{c}$ increases.
This is because a larger $s_{c}$ means that a larger part of the
prior channel support can be exploited. This simulation result also
verifies the analysis in Section IV.

\subsection{Channel Estimation Performance under Model mismatch}

In this Section, we simulate the cases of \emph{model mismatch} with
incorrect information of $s_{c}$, i.e., $\left|\mathcal{T}_{0}\bigcap\mathcal{T}\right|<s_{c}$.
Suppose that the size of shared channel support between consecutive
frames is fixed to be $|\mathcal{T}_{i}\bigcap\mathcal{T}_{i+1}|=9$,
$\forall i$, while the believed quality parameter $s_{c}$ varies
from 8 to 13 (so that the believed prior support quality is incorrect
when $s_{c}\in\{10,..,13\}$). Figure \ref{fig:NMSE-verss_sc_mismat}
illustrates the NMSE of the estimated channel versus believed quality
parameter $s_{c}$ under transmit SNR $P=25$ dB and $\bar{s}=18$.
From these figures, we observe that the performance of the M-SP degrades
severely and a larger performance degradation is observed with a larger
$s_{c}$ when $s_{c}\geq10$ (i.e., a larger model mismatch). However,
the conservative M-SP is stable and still enjoys performance gains
over the baselines. These results demonstrate the robustness of the
proposed conservative M-SP algorithm with model mismatches. 

\begin{figure}
\begin{centering}
\includegraphics[width=3.5in]{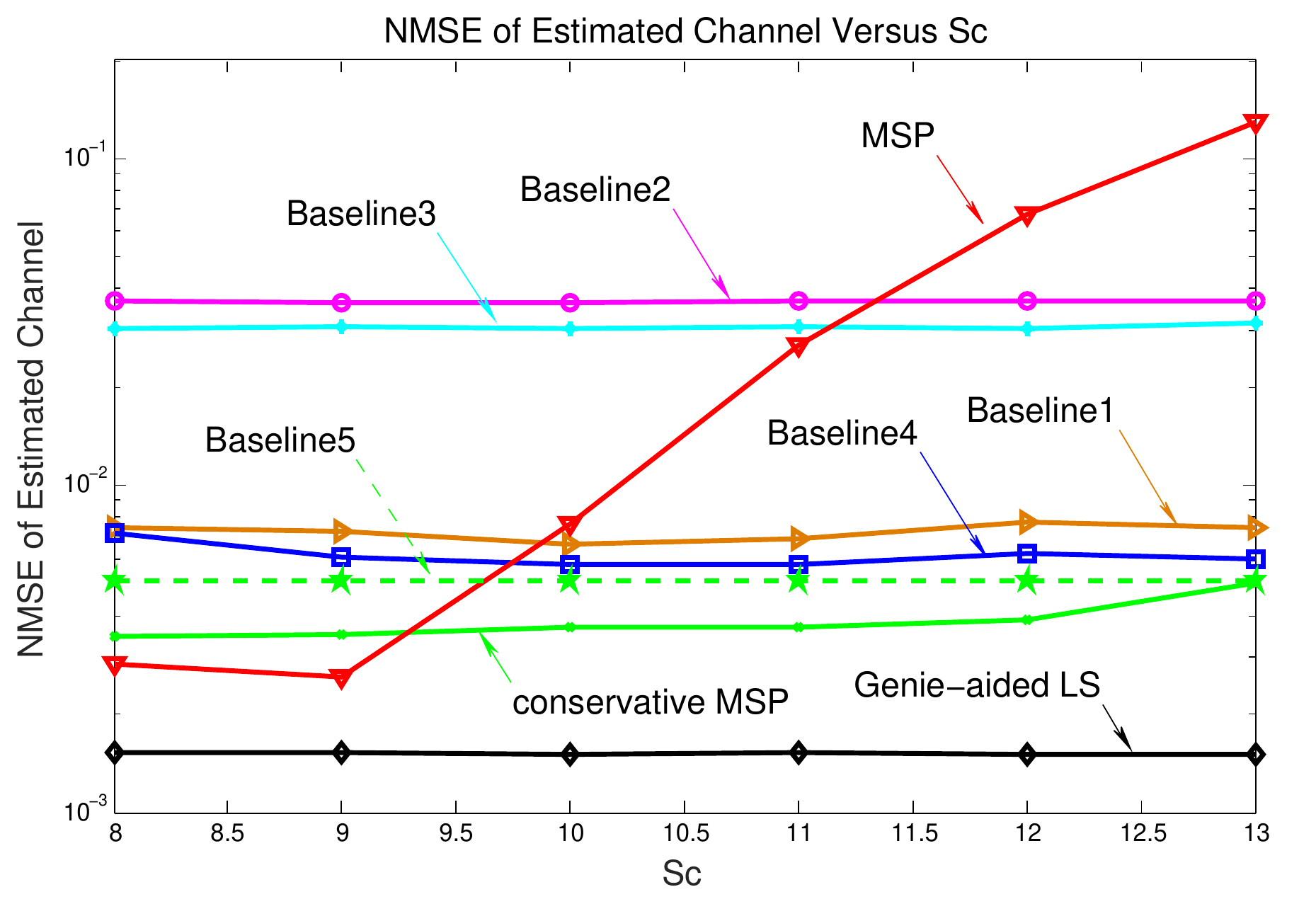}
\par\end{centering}

\protect\caption{\label{fig:NMSE-verss_sc_mismat}NMSE of estimated channel versus
the believed $s_{c}$ under model mismatch with fixed $|\mathcal{T}_{i}\bigcap\mathcal{T}_{i+1}|=9$,
$\forall i$. The other parameters are given by: transmit SNR $P=$25
dB, $T=52$ and $\bar{s}=18$.}
\end{figure}

\section{Conclusions and Future Works}

In this paper, we consider CS problems with a prior support and the
associated quality information available. Modified subspace pursuit
recovery algorithms are designed to \emph{adaptively} exploit the
prior support information to enhance the signal recovery performance.
By deploying the tools of block-RIP, we bound the recovery distortion
and we show that the proposed algorithm converges with $\mathcal{O}\left(\log\mbox{SNR}\right)$
iterations. To tolerate possible model mismatch, we have further proposed
a conservative design to have more robustness in cases of incorrect
prior support information. Finally, we apply the proposed framework
to channel estimation in massive MIMO systems with temporal correlation,
to further reduce the length of the channel training pilots. 

\appendix

\subsection{\label{sub:Proof-of-Lemma-property}Proof of Lemma \ref{Inequalities-Over-Block-RIP}}

The first two items directly follow from Definition \ref{Block-Restricted-Isometry}.
The following proves the third statement. First, we obtain $\sigma_{\max}\left(\Phi_{[\mathcal{T}_{1}\bigcup\mathcal{T}_{2}]}^{H}\Phi_{[\mathcal{T}_{1}\bigcup\mathcal{T}_{2}]}-\mathbf{I}\right)\leq\delta_{k_{1}+k_{2}|d}$
from Definition \ref{Block-Restricted-Isometry}. Second, $\Phi_{[\mathcal{T}_{1}]}^{H}\Phi_{[\mathcal{T}_{2}]}$
is a submatrix of matrix $\Phi_{[\mathcal{T}_{1}\bigcup\mathcal{T}_{2}]]}^{H}\Phi_{[\mathcal{T}_{1}\bigcup\mathcal{T}_{2}]}-\mathbf{I}$.
From the property that the spectral norm of a submatrix is always
upper bounded by the spectral norm of the entire matrix, the third
item is proved. The fourth inequality in Lemma \ref{Inequalities-Over-Block-RIP}
directly extends Lemma A.3 of \cite{chang2014improved}.

\subsection{\label{sub:Proof-of-key-lemma}Proof of Lemma \ref{Iteration-Property-II}}

We first introduce the following equalities property:

{\small{}
\begin{equation}
(\mathbf{\Phi}^{H}\mathbf{R})^{[\mathcal{T}_{1}]}=\mathbf{\Phi}_{[\mathcal{T}_{1}]}^{H}\mathbf{R}\label{eq:property1}
\end{equation}
\begin{equation}
\mathbf{\Phi}_{[\mathcal{T}_{1}]}\mathbf{X}^{[\mathcal{T}_{1}]}=\mathbf{\Phi}_{[\mathcal{T}_{1}\backslash\mathcal{T}_{2}]}\mathbf{X}^{[\mathcal{T}_{1}\backslash\mathcal{T}_{2}]}+\mathbf{\Phi}_{[\mathcal{T}_{1}\bigcap\mathcal{T}_{2}]}\mathbf{X}^{[\mathcal{T}_{1}\bigcap\mathcal{T}_{2}]}\label{eq:property2}
\end{equation}
\begin{eqnarray}
\mathbf{\Phi}_{[\mathcal{T}_{1}]}^{H}\left(\mathbf{I}-\mathbf{\Phi}_{[\mathcal{T}_{1}]}\mathbf{\Phi}_{[\mathcal{T}_{1}]}^{\dagger}\right) & = & \mathbf{0}\label{eq:property3-1}\\
\left(\mathbf{I}-\mathbf{\Phi}_{[\mathcal{T}_{1}]}\mathbf{\Phi}_{[\mathcal{T}_{1}]}^{\dagger}\right)\mathbf{\Phi}_{[\mathcal{T}_{1}]} & = & \mathbf{0}\nonumber 
\end{eqnarray}
}We first introduce the following inequalities property. Suppose $\sigma_{\min}\left(\mathbf{A}\right)$
and $\sigma_{\max}\left(\mathbf{A}\right)$ as the minimum and maximum
singular values of $\mathbf{A}\in\mathbb{C}^{M\times S}$ ($M\geq S$),
respectively, i.e., let $\mathbf{A}=U_{r}\varSigma U_{l}^{H}$, $\varSigma\in\mathbb{C}^{M\times S}$,
be the singular decomposition of $\mathbf{A}$ , $\sigma_{\min}\left(\mathbf{A}\right)=\textrm{min}\left(\mbox{diag}\left(\varSigma\right)\right)$,
$\sigma_{\max}\left(\mathbf{A}\right)=\max\left(\mbox{diag}\left(\varSigma\right)\right)$,
we have
\begin{equation}
\sigma_{\max}\left(\mathbf{A}\mathbf{B}\right)\leq\sigma_{\max}\left(\mathbf{A}\right)\sigma_{\max}\left(\mathbf{B}\right).\label{eq:Frobenius_inequality}
\end{equation}
\begin{equation}
\sigma_{\min}\left(\mathbf{A}\right)\left\Vert \mathbf{B}\right\Vert _{F}\leq\left\Vert \mathbf{A}\mathbf{B}\right\Vert _{F}\leq\sigma_{\max}\left(\mathbf{A}\right)\left\Vert \mathbf{B}\right\Vert _{F}.\label{eq:inequality2}
\end{equation}
Note that the above property (\ref{eq:property1})-(\ref{eq:inequality2})
will be frequently used in the coming proof. Based on the selection
criterion of $\mathcal{T}_{a}$, $\hat{\mathcal{T}}_{l}$, we obtain
that (i) $\mathcal{T}_{a}\triangleq\left(\mathcal{T}_{b}\bigcup\mathcal{T}_{c}\right)\bigcup\hat{\mathcal{T}}_{l}$;
(ii) $\left|\left(\mathcal{T}_{b}\bigcup\mathcal{T}_{c}\right)\right|\leq\bar{s}$,
$\left|\hat{\mathcal{T}}_{l}\right|\leq\bar{s}$, $\left|\mathcal{T}\right|\leq\bar{s}$;
(iii) Each of the three index set, i.e., $\left(\mathcal{T}_{b}\bigcup\mathcal{T}_{c}\right)$,
$\hat{\mathcal{T}}_{l}$ and $\mathcal{T}$ contain at least $s_{c}$
elements from $\mathcal{T}_{0}$; Therefore, 

{\footnotesize{}
\begin{equation}
\left|\mathcal{T}_{a}\right|\leq s_{1}\triangleq2\bar{s}+\min\left(0,\,\left|\mathcal{T}_{0}\right|-2s_{c}\right).\label{eq:first_dimen}
\end{equation}
\begin{equation}
|\mathcal{T}\bigcup\hat{\mathcal{T}}_{l}|\leq s_{1}\triangleq2\bar{s}+\min\left(0,\,\left|\mathcal{T}_{0}\right|-2s_{c}\right).\label{eq:third_dim}
\end{equation}
}{\footnotesize \par}

{\footnotesize{}
\begin{equation}
|\mathcal{T}\bigcup\mathcal{T}_{a}|\leq s_{2}\triangleq3\bar{s}+\min\left(0,\,\left|\mathcal{T}_{0}\right|-3s_{c}\right).\label{eq:second_dim}
\end{equation}
}Based on (\ref{eq:first_dimen})-(\ref{eq:second_dim}), we obtain
the following Lemma. 
\begin{lemma}
[Iteration Property]\label{Iteration-Property-proof}In the $l$-th
iteration of Algorithm \ref{alg:Modified-SP}, the following three
equations will be satisfied:
\end{lemma}
\begin{align}
\left\Vert \mathbf{R}_{(l+1)}\right\Vert _{F} & \leq\sqrt{1+\delta_{\bar{s}|d}}\left\Vert \mathbf{X}^{[\mathcal{T}\backslash\hat{\mathcal{T}}_{l+1}]}\right\Vert _{F}+\eta.\label{eq:first_core}
\end{align}
\begin{align}
\left\Vert \mathbf{X}^{[\mathcal{T}\backslash\hat{\mathcal{T}}_{l+1}]}\right\Vert _{F} & \leq\sqrt{1+\frac{4\delta_{s_{2}|d}^{2}\left(1+\delta_{s_{2}|d}\right)}{\left(1-\delta_{s_{1}|d}\right)}}\left\Vert \mathbf{X}^{[\mathcal{T}\backslash\mathcal{T}_{a}]}\right\Vert _{F}\nonumber \\
 & +\frac{2}{\sqrt{1-\delta_{s_{1}|d}}}\eta.\label{eq:second_core}
\end{align}
\begin{align}
\left\Vert \mathbf{X}^{[\mathcal{T}\backslash\mathcal{T}_{a}]}\right\Vert _{F} & \leq\frac{2\delta_{s_{2}|d}}{\left(1-\delta_{\bar{s}|d}\right)\sqrt{1-\delta_{s_{1}|d}}}\left\Vert \mathbf{R}_{(l)}\right\Vert _{F}\nonumber \\
 & +\left(\frac{2\delta_{s_{2}|d}}{\left(1-\delta_{\bar{s}|d}\right)\sqrt{1-\delta_{s_{1}|d}}}+\frac{2\sqrt{1+\delta_{\bar{s}|d}}}{1-\delta_{\bar{s}|d}}\right)\eta.\label{eq:third_core}
\end{align}

\begin{proof}
The detailed proof for equations (\ref{eq:first_core})-(\ref{eq:third_core})
are given in Appendix \ref{sub:Proof-of-equation-first}-\ref{sub:Proof-of-equation-third}
respectively. 
\end{proof}

Combine equations (\ref{eq:first_core})-(\ref{eq:third_core}), equation
(\ref{eq:Iteration_Lemma}) in Lemma \ref{Iteration-Property-II}
is derived. Next, we prove equation (\ref{eq:first_distortion}) in
Lemma \ref{Iteration-Property-II}. Based on (\ref{eq:Iteration_Lemma}),
we obtain 

{\small{}
\begin{equation}
\left\Vert \mathbf{R}_{(l)}\right\Vert _{F}\leq(C_{1})^{l}\left(\left\Vert \mathbf{R}_{(0)}\right\Vert _{F}-\frac{C_{2}\eta}{1-C_{1}}\right)+\frac{C_{2}\eta}{1-C_{1}}.\label{eq:iter_arrange}
\end{equation}
}From $\left\Vert \mathbf{R}_{(0)}\right\Vert _{F}\leq\sqrt{1+\delta_{\bar{s}|d}}\left\Vert \mathbf{X}\right\Vert _{F}+\eta$
and the fact that {\small{}
\begin{align*}
\left\Vert \mathbf{R}_{(l)}\right\Vert _{F} & =\left\Vert \Phi(\mathbf{X}-\hat{\mathbf{X}}_{(l)})+\mathbf{N}\right\Vert _{F}\\
 & \geq\sqrt{1-\delta_{s_{1}|d}}\left\Vert \mathbf{X}-\hat{\mathbf{X}}_{(l)}\right\Vert _{F}-\eta,
\end{align*}
}equation (\ref{eq:first_distortion}) is derived.

\subsection{\label{sub:Proof-of-equation-first}Proof of equation (\ref{eq:first_core})}

From the expression of $\mathbf{R}_{(l+1)}$ in \textbf{Step 2E} of
Algorithm \ref{alg:Modified-SP}, we obtain
\begin{align}
\mathbf{R}_{(l+1)} & =\left(\mathbf{I}-\Phi_{[\hat{\mathcal{T}}_{l+1}]}\Phi_{[\hat{\mathcal{T}}_{l+1}]}^{\dagger}\right)\left(\Phi_{[\mathcal{T}]}\mathbf{X}^{[\mathcal{T}]}+\mathbf{N}\right)\nonumber \\
 & =\left(\mathbf{I}-\Phi_{[\hat{\mathcal{T}}_{l+1}]}\Phi_{[\hat{\mathcal{T}}_{l+1}]}^{\dagger}\right)\Phi_{[\mathcal{T}\backslash\hat{\mathcal{T}}_{l+1}]}\mathbf{X}^{[\mathcal{T}\backslash\hat{\mathcal{T}}_{l+1}]}\nonumber \\
 & +\left(\mathbf{I}-\Phi_{[\hat{\mathcal{T}}_{l+1}]}\Phi_{[\hat{\mathcal{T}}_{l+1}]}^{\dagger}\right)\mathbf{N}\label{eq:important_equation}
\end{align}
Note that $\mathbf{I}-\Phi_{[\hat{\mathcal{T}}_{l+1}]}\Phi_{[\hat{\mathcal{T}}_{l+1}]}^{\dagger}$
is a projection matrix hence $\sigma_{\max}\left(\mathbf{I}-\Phi_{[\hat{\mathcal{T}}_{l+1}]}\Phi_{[\hat{\mathcal{T}}_{l+1}]}^{\dagger}\right)\leq1$.
From (\ref{eq:important_equation}), using properties in Lemma \ref{Inequalities-Over-Block-RIP},
equation (\ref{eq:first_core}) is proved.

\subsection{\label{sub:Proof-of-equation-second}Proof of equation (\ref{eq:second_core})}

From the selection criterion of Step 2.C in Algorithm \ref{alg:Modified-SP},
we obtain $\left\Vert \mathbf{Z}^{[\hat{\mathcal{T}}_{l+1}]}\right\Vert \geq\left\Vert \mathbf{Z}^{[\mathcal{T}]}\right\Vert $,
which leads to 
\begin{equation}
\left\Vert \mathbf{Z}^{[\mathcal{T}_{a}\backslash\mathcal{T}]}\right\Vert _{F}\geq\left\Vert \mathbf{Z}^{[\mathcal{T}_{a}\backslash\hat{\mathcal{T}}_{l+1}]}\right\Vert _{F}.\label{eq:revision_2_1}
\end{equation}
Denote $\mathbf{P}_{(\mathcal{T}_{a})}\triangleq\Phi_{[\mathcal{T}_{a}]}\left(\Phi_{[\mathcal{T}_{a}]}^{H}\Phi_{[\mathcal{T}_{a}]}\right)^{-1}\Phi_{[\mathcal{T}_{a}]}^{H}$.
We further obtain
\begin{align}
 & \mathbf{Z}^{[\mathcal{T}_{a}]}=\Phi_{[\mathcal{T}_{a}]}^{\dagger}\mathbf{Y}=\Phi_{[\mathcal{T}_{a}]}^{\dagger}\mathbf{P}_{(\mathcal{T}_{a})}\mathbf{Y}\nonumber \\
= & \Phi_{[\mathcal{T}_{a}]}^{\dagger}\mathbf{P}_{(\mathcal{T}_{a})}\left(\Phi_{[\mathcal{T}_{a}]}\mathbf{X}_{[\mathcal{T}_{a}]}+\Phi_{[\mathcal{T}\backslash\mathcal{T}_{a}]}\mathbf{X}_{[\mathcal{T}\backslash\mathcal{T}_{a}]}+\mathbf{N}\right)\nonumber \\
\triangleq & \mathbf{X}^{[\mathcal{T}_{a}]}+\mathbf{E}^{[\mathcal{T}_{a}]}+\Phi_{[\mathcal{T}_{a}]}^{\dagger}\mathbf{P}_{(\mathcal{T}_{a})}\mathbf{N}\label{eq:revision_2_2}
\end{align}
where $\mathbf{E}\in\mathbb{C}^{N\times L}$ is given by $\mathbf{E}_{[\{1,...,K\}\backslash\mathcal{T}_{a}]}=\mathbf{0}$,
$\mathbf{P}_{(\mathcal{T}_{a})}\Phi_{[\mathcal{T}\backslash\mathcal{T}_{a}]}\mathbf{X}_{[\mathcal{T}\backslash\mathcal{T}_{a}]}\triangleq\Phi_{[\mathcal{T}_{a}]}\mathbf{E}_{[\mathcal{T}_{a}]}$.
From equation (\ref{eq:revision_2_2}), we obtain 
\begin{equation}
\left\Vert \mathbf{Z}^{[\mathcal{T}_{a}\backslash\mathcal{T}]}\right\Vert _{F}\leq\left\Vert \mathbf{E}^{[\mathcal{T}_{a}]}\right\Vert _{F}+\frac{1}{\sqrt{1-\delta_{s_{1}|d}}}\eta.\label{eq:revision_2_3}
\end{equation}
\begin{align}
\left\Vert \mathbf{Z}^{[\mathcal{T}_{a}\backslash\hat{\mathcal{T}}_{l+1}]}\right\Vert _{F} & \geq\sqrt{\left\Vert \mathbf{X}^{[\mathcal{T}\backslash\hat{\mathcal{T}}_{l+1}]}\right\Vert _{F}^{2}-\left\Vert \mathbf{X}^{[\mathcal{T}\backslash\mathcal{T}_{a}]}\right\Vert _{F}^{2}}\nonumber \\
 & -\left\Vert \mathbf{E}^{[\mathcal{T}_{a}]}\right\Vert _{F}-\frac{1}{\sqrt{1-\delta_{s_{1}|d}}}\eta.\label{eq:revision_2_4}
\end{align}
We further obtain
\begin{align}
 & \sqrt{1-\delta_{s_{1}|d}}\left\Vert \mathbf{E}^{[\mathcal{T}_{a}]}\right\Vert _{F}\leq\left\Vert \Phi_{[\mathcal{T}_{a}]}\mathbf{E}^{[\mathcal{T}_{a}]}\right\Vert _{F}\nonumber \\
= & \left\Vert \mathbf{P}_{(\mathcal{T}_{a})}\Phi_{[\mathcal{T}\backslash\mathcal{T}_{a}]}\mathbf{X}^{[\mathcal{T}\backslash\mathcal{T}_{a}]}\right\Vert _{F}\overset{(a)}{\leq}\delta_{s_{2}|d}\sqrt{1+\delta_{s_{2}|d}}\left\Vert \mathbf{X}^{[\mathcal{T}\backslash\mathcal{T}_{a}]}\right\Vert _{F}\label{eq:revision_2_5}
\end{align}
where $(a)$ comes from the fourth property in Lemma \ref{Inequalities-Over-Block-RIP}.
Equation (\ref{eq:revision_2_5}) leads to 
\begin{equation}
\left\Vert \mathbf{E}^{[\mathcal{T}_{a}]}\right\Vert _{F}\leq\frac{\delta_{s_{2}|d}\sqrt{1+\delta_{s_{2}|d}}}{\sqrt{1-\delta_{s_{1}|d}}}\left\Vert \mathbf{X}^{[\mathcal{T}\backslash\mathcal{T}_{a}]}\right\Vert _{F}\label{eq:revision_2_6}
\end{equation}
Combining equation (\ref{eq:revision_2_1}), (\ref{eq:revision_2_3}),
(\ref{eq:revision_2_4}) and (\ref{eq:revision_2_6}), we obtain equation
(\ref{eq:second_core}).

\subsection{\label{sub:Proof-of-equation-third}Proof of equation (\ref{eq:third_core})}

At the beginning of the $l$-th iteration, the residue matrix $\mathbf{R}_{(l)}$
can be expressed as

\begin{align}
\mathbf{R}_{(l)} & =\left(\mathbf{I}-\Phi_{[\hat{\mathcal{T}}_{l}]}\Phi_{[\hat{\mathcal{T}}_{l}]}^{\dagger}\right)\left(\Phi_{[\mathcal{T}]}\mathbf{X}^{[\mathcal{T}]}+\mathbf{N}\right)\nonumber \\
 & \triangleq\left[\begin{array}[t]{cc}
\Phi_{[\mathcal{T}\backslash\hat{\mathcal{T}}_{l}]} & \Phi_{[\hat{\mathcal{T}}_{l}]}\end{array}\right]\tilde{\mathbf{X}}+\left(\mathbf{I}-\Phi_{[\hat{\mathcal{T}}_{l}]}\Phi_{[\hat{\mathcal{T}}_{l}]}^{\dagger}\right)\mathbf{N}\label{eq:equation_R}
\end{align}
where {\small{}$\tilde{\mathbf{X}}=\left[\begin{array}[t]{c}
\begin{array}{c}
\mathbf{X}^{[\mathcal{T}\backslash\hat{\mathcal{T}}_{l}]}\\
\mathbf{X}_{\triangle}
\end{array}\end{array}\right]$ }and $\mathbf{X}_{\triangle}=\mathbf{X}^{[\hat{\mathcal{T}}_{l}]}-\Phi_{[\hat{\mathcal{T}}_{l}]}^{\dagger}\Phi_{[\mathcal{T}]}\mathbf{X}^{[\mathcal{T}]}$.
From the properties in Lemma \ref{Inequalities-Over-Block-RIP}, equation
(\ref{eq:equation_R}) and equation (\ref{eq:third_dim}) we obtain
\begin{equation}
\left\Vert \mathbf{R}_{(l)}\right\Vert _{F}\geq\sqrt{1-\delta_{s_{1}|d}}\left\Vert \tilde{\mathbf{X}}\right\Vert _{F}-\eta\label{eq:get1}
\end{equation}
We further have the following equation {\small{}
\begin{equation}
\left\Vert \mathbf{X}^{[\mathcal{T}\backslash\mathcal{T}_{a}]}\right\Vert _{F}\leq\frac{2\delta_{s_{2}|d}}{\left(1-\delta_{\bar{s}|d}\right)}\left\Vert \tilde{\mathbf{X}}\right\Vert _{F}+\frac{2\sqrt{1+\delta_{\bar{s}|d}}}{1-\delta_{\bar{s}|d}}\eta.\label{eq:property3}
\end{equation}
}Note that equation (\ref{eq:third_core}) will be proved by combining
equation (\ref{eq:get1}) with (\ref{eq:property3}). Therefore, we
only need to prove (\ref{eq:property3}) in the following. Since both
$\left(\mathcal{T}_{c}\bigcup\mathcal{T}_{b}\right)$ and $\mathcal{T}$
contain $s_{c}$ chunks in $\mathcal{T}_{0}$, from the selection
rule of \textbf{Step 2. A}, we have 
\begin{equation}
\left\Vert \Phi_{[\mathcal{T}_{c}\bigcup\mathcal{T}_{b}]}^{H}\mathbf{R}_{(l)}\right\Vert _{F}\geq\left\Vert \Phi_{[\mathcal{T}]}^{H}\mathbf{R}_{(l)}\right\Vert _{F}\label{eq:selection_A}
\end{equation}
which derives $\left\Vert \Phi_{[(\mathcal{T}_{c}\bigcup\mathcal{T}_{b})\backslash\mathcal{T}]}^{H}\mathbf{R}_{(l)}\right\Vert _{F}\geq\left\Vert \Phi_{[\mathcal{T}\backslash(\mathcal{T}_{c}\bigcup\mathcal{T}_{b})]}^{H}\mathbf{R}_{(l)}\right\Vert _{F}$.
From this and the fact that $\Phi_{[\hat{\mathcal{T}}_{l}]}^{H}\mathbf{R}_{(l)}=\mathbf{0}$,
we further obtain{\small{}
\begin{equation}
\left\Vert \Phi_{[(\mathcal{T}_{c}\bigcup\mathcal{T}_{b})\backslash(\mathcal{T}\bigcup\hat{\mathcal{T}}^{l})]}^{H}\mathbf{R}_{(l)}\right\Vert _{F}\geq\left\Vert \Phi_{[\mathcal{T}\backslash\mathcal{T}_{a}]}^{H}\mathbf{R}_{(l)}\right\Vert _{F}\label{eq:derive2}
\end{equation}
}The right hand side term in (\ref{eq:derive2}) is further bounded
by
\begin{align}
 & \left\Vert \Phi_{[(\mathcal{T}_{c}\bigcup\mathcal{T}_{b})\backslash(\mathcal{T}\bigcup\hat{\mathcal{T}}^{l})]}^{H}\mathbf{R}_{(l)}\right\Vert _{F}\nonumber \\
= & \left\Vert \Phi_{[(\mathcal{T}_{c}\bigcup\mathcal{T}_{b})\backslash(\mathcal{T}\bigcup\hat{\mathcal{T}}^{l})]}^{H}\times\vphantom{\left(\left[\begin{array}[t]{cc}
\Phi_{[\mathcal{T}\backslash\hat{\mathcal{T}}^{l}]} & \Phi_{[\hat{\mathcal{T}}^{l}]}\end{array}\right]\tilde{\mathbf{X}}+\left(\mathbf{I}-\Phi_{[\hat{\mathcal{T}}^{l}]}\Phi_{[\hat{\mathcal{T}}^{l}]}^{\dagger}\right)\mathbf{N}\right)}\right.\nonumber \\
 & \left.\left(\left[\begin{array}[t]{cc}
\Phi_{[\mathcal{T}\backslash\hat{\mathcal{T}}_{l}]} & \Phi_{[\hat{\mathcal{T}}_{l}]}\end{array}\right]\tilde{\mathbf{X}}+\left(\mathbf{I}-\Phi_{[\hat{\mathcal{T}}_{l}]}\Phi_{[\hat{\mathcal{T}}_{l}]}^{\dagger}\right)\mathbf{N}\right)\right\Vert _{F}\nonumber \\
\leq & \delta_{s_{2}|d}\left\Vert \tilde{\mathbf{X}}\right\Vert _{F}+\sqrt{1+\delta_{\bar{s}|d}}\eta\label{eq:second1}
\end{align}
The left hand side term in equation (\ref{eq:derive2}) is further
bounded by

{\small{}
\begin{align}
 & \left\Vert \Phi_{[\mathcal{T}\backslash\mathcal{T}_{a}]}^{H}\mathbf{R}_{(l)}\right\Vert _{F}\nonumber \\
= & \left\Vert \Phi_{[\mathcal{T}\backslash\mathcal{T}_{a}]}^{H}\times\vphantom{\left(\left[\begin{array}[t]{cc}
\Phi_{[\mathcal{T}\backslash\hat{\mathcal{T}}^{l}]} & \Phi_{[\hat{\mathcal{T}}^{l}]}\end{array}\right]\tilde{\mathbf{X}}+\left(\mathbf{I}-\Phi_{[\hat{\mathcal{T}}^{l}]}\Phi_{[\hat{\mathcal{T}}^{l}]}^{\dagger}\right)\mathbf{N}\right)}\right.\nonumber \\
= & \left.\left(\left[\begin{array}[t]{cc}
\Phi_{[\mathcal{T}\backslash\hat{\mathcal{T}}_{l}]} & \Phi_{[\hat{\mathcal{T}}_{l}]}\end{array}\right]\tilde{\mathbf{X}}+\left(\mathbf{I}-\Phi_{[\hat{\mathcal{T}}_{l}]}\Phi_{[\hat{\mathcal{T}}_{l}]}^{\dagger}\right)\mathbf{N}\right)\right\Vert _{F}\nonumber \\
\geq & \left\Vert \Phi_{[\mathcal{T}\backslash\mathcal{T}_{a}]}^{H}\left(\left[\begin{array}[t]{cc}
\Phi_{[\mathcal{T}\backslash\mathcal{T}_{a}]} & \Phi_{[\mathcal{T}_{a}]}\end{array}\right]\right)\left[\begin{array}[t]{c}
\begin{array}{c}
\mathbf{X}^{[\mathcal{T}\backslash\mathcal{T}_{a}]}\\
\mathbf{X}_{\triangle}^{'}
\end{array}\end{array}\right]\right\Vert _{F}\nonumber \\
 & -\sqrt{1+\delta_{\bar{s}|d}}\eta\label{eq:second2}
\end{align}
}where $\mathbf{X}_{\triangle}^{'}$ is obtained by rewritten $\left[\begin{array}[t]{cc}
\Phi_{[\mathcal{T}\backslash\hat{\mathcal{T}}_{l}]} & \Phi_{[\hat{\mathcal{T}}_{l}]}\end{array}\right]\tilde{\mathbf{X}}$ to be $\left(\left[\begin{array}[t]{cc}
\Phi_{[\mathcal{T}\backslash\mathcal{T}_{a}]} & \Phi_{[\mathcal{T}_{a}]}\end{array}\right]\right)\left[\begin{array}[t]{c}
\begin{array}{c}
\mathbf{X}^{[\mathcal{T}\backslash\mathcal{T}_{a}]}\\
\mathbf{X}_{\triangle}^{'}
\end{array}\end{array}\right]$. Note that this $\mathbf{X}_{\triangle}^{'}$ can always be found
because $\hat{\mathcal{T}}_{l}\subseteq\mathcal{T}_{a}$. Furthermore
$||\mathbf{X}_{\triangle}^{'}||_{F}\leq\left\Vert \tilde{\mathbf{X}}\right\Vert _{F}$.
Continuing the derivation in (\ref{eq:second2}), we obtain

{\small{}
\begin{align}
 & \left\Vert \Phi_{[\mathcal{T}\backslash\mathcal{T}_{a}]}^{H}\mathbf{R}_{(l)}\right\Vert _{F}\nonumber \\
\geq & \sigma_{\min}\left(\Phi_{[\mathcal{T}\backslash\mathcal{T}_{a}]}^{H}\Phi_{[\mathcal{T}\backslash\mathcal{T}_{a}]}\right)\left\Vert \mathbf{X}^{[\mathcal{T}\backslash\mathcal{T}_{a}]}\right\Vert _{F}\nonumber \\
 & -\sigma_{\textrm{max}}\left(\Phi_{[\mathcal{T}\backslash\mathcal{T}_{a}]}^{H}\Phi_{[\mathcal{T}_{a}]}\right)\left\Vert \mathbf{X}_{\triangle}^{'}\right\Vert _{F}-\sqrt{1+\delta_{\bar{s}|d}}\eta\nonumber \\
\geq & (1-\delta_{\bar{s}|d})\left\Vert \mathbf{X}^{[\mathcal{T}\backslash\mathcal{T}_{a}]}\right\Vert _{F}-\delta_{s_{2}|d}\left\Vert \tilde{\mathbf{X}}\right\Vert _{F}-\sqrt{1+\delta_{\bar{s}|d}}\eta.\label{eq:second3}
\end{align}
}{\small \par}

Combine the results in equation (\ref{eq:derive2}), (\ref{eq:second1})
and (\ref{eq:second3}), we obtain{\small{}
\begin{align*}
 & \delta_{s_{2}|d}\left\Vert \tilde{\mathbf{X}}\right\Vert _{F}+\sqrt{1+\delta_{\bar{s}|d}}\eta\geq\\
 & (1-\delta_{\bar{s}|d})\left\Vert \mathbf{X}^{[\mathcal{T}\backslash\mathcal{T}_{a}]}\right\Vert _{F}-\delta_{s_{2}|d}\left\Vert \tilde{\mathbf{X}}\right\Vert _{F}-\sqrt{1+\delta_{\bar{s}|d}}\eta
\end{align*}
} which further derives the desired equation (\ref{eq:third_core}).

\subsection{\label{sub:Proof-of-Theorem-Distortion_Bound}Proof of Theorem \ref{Distortion-BoundDenote-MSP}}

If Algorithm \ref{alg:Modified-SP} stops from the condition of $\left\Vert \mathbf{R}_{(l+1)}\right\Vert _{F}\leq\gamma$,
then the obtained solution $\mathbf{\hat{X}}=\hat{\mathbf{X}}_{(l+1)}$
and we obtain $\left\Vert \mathbf{X}-\hat{\mathbf{X}}\right\Vert _{F}\leq\frac{\gamma+\eta}{\sqrt{1-\delta_{s_{1}|d}}}$.
If Algorithm \ref{alg:Modified-SP} stops from the condition of $\left\Vert \mathbf{R}_{(l+1)}\right\Vert _{F}\geq\left\Vert \mathbf{R}_{(l)}\right\Vert _{F}$,
then obtained solution $\mathbf{\hat{X}}=\hat{\mathbf{X}}_{(l)}$.
From equation (\ref{eq:Iteration_Lemma}), we obtain {\small{}
\[
\left\Vert \mathbf{R}_{(l)}\right\Vert _{F}\leq\left\Vert \mathbf{R}_{(l+1)}\right\Vert _{F}\leq C_{1}\left\Vert \mathbf{R}_{(l)}\right\Vert _{F}+C_{2}\eta.
\]
}From $\delta_{s_{2}|d}<0.246$, we obtain $C_{1}<1$ and $\left\Vert \mathbf{R}_{(l)}\right\Vert _{F}\leq\frac{C_{2}\eta}{1-C_{1}}$.
Further from $\mathbf{R}_{(l)}\geq\sqrt{1-\delta_{s_{1}|d}}\left\Vert \mathbf{X}-\hat{\mathbf{X}}\right\Vert _{F}-\eta$,
we obtain $\left\Vert \mathbf{X}-\hat{\mathbf{X}}\right\Vert _{F}\leq C_{4}\eta$.
Hence equation (\ref{eq:Distortion_Bound}) is proved. Next, we prove
(\ref{eq:refined_bound}). Note that when {\small{}$\min_{k\in\mathcal{T}}\left\Vert \mathbf{X}[k]\right\Vert _{F}>\max\left(C_{4}\eta,\quad\frac{\gamma+\eta}{\sqrt{1-\delta_{s_{1}|d}}}\right)$},
the identified signal support $\mathcal{\hat{T}}$ must be correct,
i.e., $\mathcal{T}\subseteq\mathcal{\hat{T}}$. This can be proved
via the contradiction method (i.e., $\exists i\in\mathcal{T}$, $i\notin\mathcal{\hat{T}}$.
We obtain {\small{}$\left\Vert \mathbf{X}-\hat{\mathbf{X}}\right\Vert _{F}\geq\left\Vert \mathbf{X}[i]\right\Vert _{F}>\max\left(C_{4}\eta,\quad\frac{\gamma+\eta}{\sqrt{1-\delta_{s_{1}|d}}}\right)$},
which violates equation (\ref{eq:Distortion_Bound})). From $\mathcal{T}\subseteq\mathcal{\hat{T}}$,
we obtain 
\[
\mathbf{X}-\hat{\mathbf{X}}=\mathbf{X}_{[\mathcal{\hat{T}}]}-\Phi_{[\mathcal{\hat{T}}]}^{\dagger}\left(\Phi_{[\mathcal{\hat{T}}]}^{\dagger}\mathbf{X}+\mathbf{N}\right)=-\Phi_{[\mathcal{\hat{T}}]}^{\dagger}\mathbf{N}
\]
which further derives (\ref{eq:refined_bound}) from Lemma \ref{Inequalities-Over-Block-RIP}.

\subsection{\label{sub:Proof-of-Theorem-con}Proof of Theorem \ref{Convergence-SpeedSuppose-II}}

First, $\left\Vert \mathbf{R}_{(0)}\right\Vert _{F}\leq\sqrt{1+\delta_{\bar{s}|d}}\rho^{\frac{1}{2}}+\eta$.
Second, from (\ref{eq:iter_arrange}), after $n$ iterations in Step
2 of Algorithm \ref{alg:Modified-SP}, the following inequality hold:{\small{}
\begin{align}
\left\Vert \mathbf{R}_{(n)}\right\Vert _{F} & \leq\frac{C_{2}\eta}{1-C_{1}}+(C_{1})^{n}\left(\sqrt{1+\delta_{\bar{s}|d}}\rho^{\frac{1}{2}}+\eta-\frac{C_{2}\eta}{1-C_{1}}\right)\label{eq:iter_arrange2}
\end{align}
}From (\ref{eq:iter_arrange2}), when $\gamma>\frac{C_{2}\eta}{1-C_{1}}$
, $n=n_{co}$, we must obtain $\left\Vert \mathbf{R}_{(n)}\right\Vert _{F}\leq\gamma$
and hence Step 2 of Algorithm \ref{alg:Modified-SP} must have stopped
after $n_{co}$ iterations, where $n_{co}$ is as given in Theorem
\ref{Convergence-SpeedSuppose-II}).

\subsection{\label{sub:Proof-of-Theorem-conservative_MSP}Proof of Theorem \ref{Distortion-Bound-of-conservative_MSP}}

Note that for the conservative M-SP, we have (i) $\left|\mathcal{T}_{a}\right|\leq2\bar{s}+s_{c}$,
$\left|\mathcal{T}_{a}\bigcup\mathcal{T}\right|\leq s_{3}\triangleq3\bar{s}+\min\left(s_{c},\left|\mathcal{T}_{0}\right|-\left|\mathcal{T}_{0}\bigcap\mathcal{T}\right|\right)$,
$|\mathcal{T}\bigcup\hat{\mathcal{T}}_{l}|\leq2\bar{s}$; (ii) equation
(\ref{eq:selection_A}), (\ref{eq:revision_2_1}) for Step 2A and
Step 2C, respectively, will hold no matter whether the quality information
$s_{c}$ is correct or not. Following the proof of Appendix \ref{sub:Proof-of-key-lemma},
we would obtain the following iteration property for the conservative
M-SP,
\begin{equation}
\left\Vert \mathbf{R}_{(l+1)}\right\Vert _{F}\leq C_{5}\left\Vert \mathbf{R}_{(l)}\right\Vert _{F}+C_{6}\eta\label{eq:key_iteration_cons_MSP}
\end{equation}
where $C_{5}$, $C_{6}$ are modified correspondingly (compared with
their counterpart $C_{1}$, $C_{2}$) and are given in Table \ref{tab:The-detailed-expressions-constants-conservative}.
On the other hand, if Algorithm \ref{alg:Modified-SP} stops from
the condition of $\left\Vert \mathbf{R}_{(l+1)}\right\Vert _{F}\leq\gamma$,
then the obtained solution $\mathbf{\hat{X}}=\hat{\mathbf{X}}_{(l+1)}$
satisfies $\left\Vert \mathbf{X}-\hat{\mathbf{X}}\right\Vert _{F}\leq\frac{\gamma+\eta}{\sqrt{1-\delta_{2\bar{s}|d}}}$
similar to Appendix \ref{sub:Proof-of-Theorem-Distortion_Bound}.
If Algorithm \ref{alg:Modified-SP} stops from the condition of $\left\Vert \mathbf{R}_{(l+1)}\right\Vert _{F}\geq\left\Vert \mathbf{R}_{(l)}\right\Vert _{F}$,
from (\ref{eq:key_iteration_cons_MSP}), we obtain $\left\Vert \mathbf{R}_{(l)}\right\Vert _{F}\leq\frac{C_{6}\eta}{1-C_{5}}$.
Furthermore, from $\mathbf{R}_{(l)}\geq\sqrt{1-\delta_{2\bar{s}|d}}\left\Vert \mathbf{X}-\hat{\mathbf{X}}\right\Vert _{F}-\eta$,
we obtain $\left\Vert \mathbf{X}-\hat{\mathbf{X}}\right\Vert _{F}\leq C_{7}\eta$.
Subsequently, equation (\ref{eq:dist_bound-conservative}) in Theorem
\ref{Distortion-Bound-of-conservative_MSP} is proved. Based on (\ref{eq:dist_bound-conservative}),
equation (\ref{eq:refined_bound-conservative}) can be obtained similar
to (\ref{eq:refined_bound}).

\bibliographystyle{IEEEtran}
\bibliography{CSIT_REF}

\end{document}